\documentclass[a4paper,10pt]{article}
\usepackage{amssymb,amsbsy,amsmath,amsfonts,amscd,amsthm}
\usepackage[T1]{fontenc}
\usepackage{textcomp}
\usepackage{graphicx}
\usepackage{color}
\usepackage{pdfsync}
\usepackage{enumitem}
\usepackage{hyperref}

\newtheorem{theorem}{Theorem}[section]

\newtheorem{proposition}[theorem]{Proposition}

\newtheorem{definition}[theorem]{Definition}
\theoremstyle{remark}
\newtheorem{remark}[theorem]{Remark}

\def\R{{\mathbb R}}
\def\M{{\mathbb M}}

\def\cof{\mathop{\textnormal{cof}}}
\def\tr{\mathop{\textnormal{tr}}}
\def\div{\mathop{\textnormal{div}}}

\def\Rho{\mathrm{P}}
\def\SO{\mathrm{SO}}
\def\SL{\mathrm{SL}}
\def\Sym{\mathrm{Sym}}
\def\Skew{\mathrm{Skew}}
\def\fluide{\mathrm{fl}}
\title{A few remarks on thermomechanics}
\author{Herv\'e Le Dret$^1$ and Annie Raoult$^2$}{}

\begin{document}
\maketitle

% Enter the first author's name and email address; email addresses are required for each author.
% Use footnote notations to indicate address and affiliations with commas between numbers if more than one address applies; see below for examples.
%\centerline{\scshape
%Herv\'e Le Dret$^{{\href{mailto:herve.le\_dret@sorbonne-universite.fr}}\,*1}$
%and Annie Raoult$^{{\href{mailto:annie.raoult@parisdescartes.fr}}\,2}$}

\medskip

{\footnotesize
% Enter the full affiliation and country name:
% Do not insert commas or periods at the end of lines.
 \centerline{$^1$Sorbonne Universit\'e, Universit\'e Paris Cit\'e, CNRS, Laboratoire Jacques-Louis Lions, F-75005 Paris, France}
} % Do not forget to end {\footnotesize with the sign }

\medskip

{\footnotesize
 % Enter the full affiliation and country name:
 \centerline{$^2$Universit\'e Paris Cit\'e, CNRS, MAP5, F-75006 Paris, France}
}

\bigskip

%%%%%%%%%%%%%%%%%%%%%%%%%%%%%%%%%%%%%%%%%%%%%%%%%%%%%%%
%             5. ABSTRACT
%%%%%%%%%%%%%%%%%%%%%%%%%%%%%%%%%%%%%%%%%%%%%%%%%%%%%%%

\begin{abstract}
We develop a global setting for modeling thermo-visco-elastic materials that satisfy the principles of thermodynamics and are properly invariant. This setting encompasses many known solid and fluid models, as well as new models with internal variables that generalize the Maxwell rheological model. Complex fluid models such as the Oldroyd B model are shown to belong to the above general family of models. The specific Oldroyd B model is however found to be seriously lacking in terms of satisfying the second principle of thermodynamics. On the contrary, a second  complex fluid model based on the Zaremba-Jaumann derivative is shown to satisfy the second principle of thermodynamics.
\end{abstract}

%%%%%%%%%%%%%%%%%%%%%%%%%%%%%%%%%%%%%%%%%%%%%%%%%%%%%%
%                   6. BODY
%%%%%%%%%%%%%%%%%%%%%%%%%%%%%%%%%%%%%%%%%%%%%%%%%%%%%%

% Only the first word and proper nouns of section titles should be capitalized.
% The title of section 1:
\section{Introduction}
Following a long tradition that started with \cite{Truesdell-Noll} and \cite{Coleman-Mizel} and that is still active to this day, we present in this article a unifying and comprehensive framework for thermo-visco-elastic material models with internal variables that first and foremost obey the principles of thermodynamics. We put special emphasis on obtaining three-dimensional, nonlinear and properly invariant models that cover a wide spectrum of material behavior.

After a brief review of notation, basic thermodynamical concepts and proper invariance principles, we introduce in Section \ref{le truc principal} an extended set of thermodynamic variables that includes the velocity gradient and two kinds of internal variables. The internal variables denoted by $\Xi$ are governed by an ordinary differential equation and play a crucial role in questions related to the second principle and dissipation, whereas the internal variables denoted by $\Pi$ are featured for purposes of generality, even though their role is less crucial in terms of thermodynamics. This extended set of thermodynamic variables does not seem to appear all at once in the literature, as far as we are aware.

We then perform the classical Coleman-Noll procedure in complete detail, based on assumed constitutive laws in the above thermodynamic variables for the first Piol\`a-Kirchhoff stress, the heat flux, the entropy and the Helmholtz free energy, complemented by a flow rule for the internal variables $\Xi$. We thus obtain a set of conditions on these constitutive laws that are  necessary and sufficient for the second principle of thermodynamics to be satisfied by this class of materials.  These conditions extend well known conditions such as the relationship between the entropy and the free energy, or the decomposition of the stress into a thermoelastic part and a dissipative part, see Proposition \ref{restrictions thermo avec variables internes}. We show in particular that the free energy does not depend on the velocity and temperature gradients, nor on the internal variables $\Pi$. We then discuss conditions under which the Clausius-Duhem inequality is equivalent to the Clausius-Planck inequalities for the materials under consideration.

We introduce next dissipation potentials in all thermodynamic variables that yield pairs of constitutive laws for the dissipative part of the stress and for the flow rule that satisfy the mechanical part of the Clausius-Planck inequalities by design. 

We then turn in Section \ref{sursection AIM visqueux} to frame-indifference and material symmetries issues, when there are no internal variables. We characterize all constitutive laws for the stress tensor expressed in terms of the deformation and velocity gradients (plus temperature) that are frame-indifferent. This characterization does not seem to be widely known, see Proposition \ref{AIM lagrangien visqueux}. We also focus on frame-indifference for dissipation potentials and its consequences. 

 Next, we analyze frame-indifference  for the heat flux and thermal symmetries. In particular, in Proposition \ref{chaleur isotrope}, we give a complete description of isotropic heat fluxes that is more comprehensive than the results found in the literature, which generally assume unnecessarily stringent invariance conditions. We also describe all fluid heat fluxes. 

In the remainder of the article, Section \ref{des exemples}, we consider several very different examples. These obviously include elastic and thermo-elastic solids and fluids as well as kinematically viscous solids and fluids (Reiner-Rivlin fluids, compressible Newtonian fluids). We also introduce in Section \ref{modeles maxwell 3d} a new family of visco-elastic materials with internal variables that are generalizations of the classical Maxwell rheological model, which we call nonlinear 3d Maxwell models. These models are based on a free energy of the form $\widehat A_m(F,F_i)= \widehat W(FF_i^{-1})$, where $F$ is the deformation gradient, $F_i$ is an internal variable acting as an internal viscous strain and $\widehat W$ is any frame-indifferent nonlinearly elastic stored energy function. We pair this free energy with appropriate frame-indifferent dissipation potentials to produce flow rules that are compatible with the mechanical part of the Clausius-Planck inequalities. The issue of frame-indifference is nonstandard here, because of the presence of the internal variable, and requires independent developments. The same goes for material symmetries. In particular, we give specific examples of isotropic solid nonlinear 3d Maxwell materials as well as fluid nonlinear 3d Maxwell materials, which are all frame-indifferent and satisfy the second principle of thermodynamics. We finally indicate that our nonlinear 3d Maxwell materials can exhibit stress relaxation.

In a similar spirit, we also introduce nonlinear 3d frame-indifferent generalizations of the Kelvin-Voigt and generalized Maxwell rheological models. These new models are capable of exhibiting creep and both stress relaxation and creep respectively. 

Finally, in Section \ref{fluides complexes}, we propose a critical look at Oldroyd B complex fluids in the light of thermodynamical requirements. Indeed, we provide numerical evidence that the naive formulation of internal dissipation, $\sigma:d$, can take strictly negative values, which thus violate the Clausius-Planck inequalities.  The same holds true for a similar complex fluid model obtained by replacing the Oldroyd B derivative by the Zaremba-Jaumann derivative. We then proceed to show that both Oldroyd B and Zaremba-Jaumann fluid models, and similar models obtained from more general objective derivatives, can be recast as viscous fluid models with an internal variable---the polymer part of the stress—and are thus actually part of our general scheme developed above. We show that, in the Oldroyd B case, it is impossible to find a free energy that makes the corresponding dissipation nonnegative, so that the Clausius-Planck inequalities are also always violated by an Oldroyd B fluid viewed as a viscous fluid with an internal variable. On the contrary, the Zaremba-Jaumann fluid admits a very simple free energy for which the mechanical part of Clausius Planck inequality is satisfied, thus bringing to light a striking thermodynamic difference between these two superficially similar complex fluid models and casting doubts on the thermodynamic viability of the Oldroyd B model.

% The title of section 2:
\section{Thermomechanical background}This section is mainly to fix notation, everything being otherwise well known. A few  references for this material are \cite{Gurtinandco} and \cite{Truesdell-Noll}.
We use the convenient typographical convention of denoting any quantity pertaining to the Lagrangian description with an uppercase letter and the corresponding Eulerian quantity with the corresponding lowercase letter, inasmuch as possible. This rule however suffers a few exceptions.
% The title of the first subsection in section 2:
\subsection{Kinematics}
The Lagrangian description of a material body consists in considering an arbitrary reference configuration $\Omega$, which is an open subset of $\R^3$, the points of which $X$ are used to label the  material particles composing the body. Their motion is described by a deformation mapping  $(X,t)\mapsto \phi(X,t)\in E$, where $E$ is the physical Euclidean three-dimensional space. The deformed configuration at time $t$ is $\phi(\Omega,t)$. We will assume throughout that $\phi$ is sufficiently regular, invertible and orientation preserving for fixed $t$. We denote by $V(X,t)=\frac{\partial\phi}{\partial t}(X,t)$ the velocity of the particles, and by  $\Gamma(X,t)=\frac{\partial^2\phi}{\partial t^2}(X,t)$ their acceleration. The deformation gradient is denoted $F(X,t)=\nabla_X\phi(X,t)$ with Jacobian $J(X,t)=\det F(X,t)$,  and the deformation rate is $H(X,t)=\nabla_XV(X,t)=\frac{\partial F}{\partial t}(X,t)$. 

In the Eulerian description, we are interested in what actually happens in physical space-time at given points $(x,t)\in E\times\R$. The connection with the Lagrangian description happens when $(x,t)=(\phi(X,t),t)$. In this case, the velocity of particules is $v(x,t)=V(\phi^{-1}(x,t),t)$ and their acceleration
$$\gamma(x,t)=\Gamma(\phi^{-1}(x,t),t)=\frac{\partial v}{\partial t}(x,t)+\bigl(\nabla_xv(x,t)\bigr)v(x,t) .$$
This is a particular case of the more general material derivative for any differentiable function defined on space-time with values in a normed vector space $\tau\colon E\times\R\to F$, which is simply 
$$\dot\tau(x,t)=d\tau(x,t)(v(x,t),1),$$
(assuming $v$ to be defined everywhere for simplicity) so that  
$$\dot\tau=\frac{\partial}{\partial t}(\tau\circ\Phi)\text{ with }\Phi(X,t)=(\phi(X,t),t).$$
Indeed, $\gamma=\dot v$.
We also use the notation $h(x,t)=\nabla_xv(x,t)=\nabla_XV(X,t)F^{-1}(X,t)$, with the understanding that $(x,t)=\Phi(X,t)$, for the Eulerian velocity gradient, whose symmetric part $d(x,t)=\frac12\bigl(h(x,t)^T+h(x,t)\bigr)$ is the stretching tensor and antisymmetric part $w(x,t)=\frac12\bigl(h(x,t)^T-h(x,t)\bigr)$ is the spin tensor.

\subsection{Dynamics}
Any subbody $A\subset \Omega$ contains a certain amount of mass $\mathcal{M}(A)=\int_A\Rho(X)\,dX$, where the letter $\Rho$ is a capital rho, denoting a given function that represents the mass density in the Lagrangian description. It is assumed that this function does not depend on $t$, so that there is no mass transfer during deformations and mass is conserved, $\mathcal{M}(A)=\int_{\phi(A,t)}\rho(x,t)\,dx$, where $\rho$ is the actual Eulerian mass density, $\rho(\phi(X,t),t)=\frac{\Rho(X)}{J(X,t)}$. This  is equivalent to the classical mass conservation law
$$
\dot\rho+\rho \,{\div\nolimits_x}v=0\text{ in }\phi(\Omega,t).
$$

The fundamental law of dynamics or conservation of momentum applied to any subbody and the Cauchy axiom imply the existence of the Cauchy stress tensor $\sigma$ in the Eulerian description, which satisfies the dynamics equation
\begin{align*}
&\rho\dot v-{\div\nolimits_x}\sigma=b,\\
&\sigma^T=\sigma,
\end{align*}
in $\phi(\Omega,t)$, where $b$ is the applied body force density. There may or may not be additional traction conditions on part of the boundary of $\phi(\Omega,t)$. The Eulerian formulation is somewhat unwieldy since $\phi(\Omega,t)$ is unknown in general (except say for fluids filling a container). 

In this respect things look more controllable in the Lagrangian description. First of all, there is no mass conservation law. The first Piol\`a-Kirchhoff stress tensor is classically introduced as
$$
T_{\mathrm{R}}(X,t)=\sigma(\phi(X,t),t)\cof F(X,t),
$$
where $\cof F$ denotes the cofactor matrix of $F$. After pulling back the applied body force density to the reference configuration by $B(X,t)=J(X,t)b(\phi(X,t),t)$, the dynamics equation assume the Lagrangian form
\begin{align}
&\Rho\Gamma-{\div\nolimits_X}T_{\mathrm{R}}=B,\label{cons mom lag}\\
&FT_{\mathrm{R}}^T=T_{\mathrm{R}}F^T,\label{sym PK}
\end{align}
in $\Omega$.

\subsection{The first principle}
From now on, heat must be taken into account. It is assumed that there is a supplied volumic thermal power density $r$ as well as thermal power flowing through surfaces via a heat flux vector $q$ (whose existence can be deduced from an adapted version of Cauchy's axiom)  that contribute to the thermal power affecting any subbody according to 
$$\mathcal{Q}_A(t)=\int_{\phi(A,t)}r(x,t)\,dx-\int_{\partial\phi(A,t)}q(x,t)\cdot n_{\phi(A,t)}\,d\sigma$$
where $d\sigma$  denotes the surfacic measure and $n_{\phi(A,t)}$ is the unit exterior normal vector to $\phi(A,t)$. One also has to consider the internal mechanical power
$$\mathcal{P}_{A,\mathrm{int}}(t)=\int_{\phi(A,t)} \sigma(x,t):\nabla_xv(x,t)\,dx=\int_{\phi(A,t)} \sigma(x,t):d(x,t)\,dx$$
where the colon denotes the usual Frobenius inner product of $3\times 3$ matrices ($\sigma$ is symmetric thus the spin tensor $w$ does not contribute to mechanical power).

The first principle can be  stated as the fact that the internal energy $\mathcal{E}_A(t)$ in any subbody $A$ varies in time according to the sum of these powers,
\begin{equation}\label{premier principe}
\mathcal{E}'_A(t)=\mathcal{P}_{A,\mathrm{int}}(t)+\mathcal{Q}_A(t).
\end{equation}

It is assumed that this energy functional has a specific density 
$$\mathcal{E}_A(t)=\int_{\phi(A,t)}\rho(x,t)e_m(x,t)\,dx,$$ 
then equation \eqref{premier principe} is equivalent to the so-called energy equation
$$
\rho\dot e_m=\sigma:d+r-{\div\nolimits_x}q.
$$
Also note that the energy is defined up to an additive constant, only energy differences are physically significant.

To express the same relations in the Lagrangian description, we first need to perform a Piolà transform on the heat flux
$$
Q(X,t)=(\cof F(X,t))^Tq(\phi(X,t),t)
$$
and introduce the reference specific energy density $E_m$, such that 
$$\mathcal{E}_A(t)=\int_A\Rho(X) E_m(X,t)\,dX.$$
The energy equation then becomes 
\begin{equation}\label{eq d'energie ref}
\Rho\frac{\partial E_m}{\partial t}=T_{\mathrm{R}}:\nabla_XV+R-{\div\nolimits_X}Q,
\end{equation}
where $R$ denotes the Lagrangian thermal power source.

\subsection{The second principle}
The dynamics and energy equations actually are equations that will eventually govern the evolution, although they are still incomplete at this stage. The second principle is of a different nature since it precludes some of these evolutions. 

In the Eulerian description, it is assumed that there exists an absolute temperature field $\theta>0$ and that each subbody contains a certain amount of a quantity called entropy 
$\mathcal{S}_A(t)$. The second principle stipulates that the following inequality must always hold
\begin{equation}\label{Clausius-Duhem}
\mathcal{S}'_A(t)-\int_{\phi(A,t)}\frac{r(x,t)}{\theta(x,t)}\,dx+\int_{\partial\phi(A,t)}\frac{q(x,t)\cdot n_{\phi(A,t)}}{\theta(x,t)}\,d\sigma\ge 0.
\end{equation}
This inequality is known as the Clausius-Duhem inequality. Note that if the inequality is strict for a certain evolution in time, then the corresponding time-reversed evolution with time-reversed entropy and opposite time-reversed heat sources and heat fluxes cannot satisfy it. 
Such an evolution is called  irreversible.

It is also assumed that the entropy  has a specific density 
$$\mathcal{S}_A(t)=\int_{\phi(A,t)}\rho(x,t)s_m(x,t)\,dx,$$ 
then  \eqref{Clausius-Duhem} is equivalent to the  differential inequality
$$
\rho\dot s_m-\frac r\theta+\frac{{\div\nolimits_x}q}{\theta}-\frac{q\cdot\nabla_x\theta}{\theta^2}\ge0,
$$
also referred to as the Clausius-Duhem inequality. 

In the Lagrangian description, we have a specific entropy density $S_m$ such that $\mathcal{S}_A(t)=\int_{A}\Rho(X)S_m(X,t)\,dX$
and a temperature field $\Theta(X,t)=\theta(\phi(X,t),t)$ and the Clausius-Duhem inequality reads
$$
\Rho\frac{\partial S_m}{\partial t}-\frac R\Theta+\frac{{\div\nolimits_X}Q}{\Theta}-\frac{Q\cdot\nabla_X\Theta}{\Theta^2}\ge0.
$$

The Clausius-Duhem inequality is in particular satisfied when the following inequalities, known as the Clausius-Planck inequalities are satisfied
\begin{align*}
&\rho\theta\dot s_m-r+{\div\nolimits_x}q\ge0\text{ and }q\cdot\nabla_x\theta\le0,\\
&\Rho\Theta\frac{\partial S_m}{\partial t}-R+{\div\nolimits_X}Q\ge0\text{ and }Q\cdot\nabla_X\Theta\le0 .
\end{align*}
We will see later on situations in which the Clausius-Planck inequalities actually follow from the Clausius-Duhem inequality and situations where they do not.

There are many quantities called free energies in the literature. The one free energy that is the most adapted to our purposes is the Helmholtz free energy, which is defined by 
$$
a_m=e_m-\theta s_m,\qquad A_m=E_m-\Theta S_m,
$$
in both descriptions. The main advantage of the Helmholtz free energy is that it can be used to rewrite the Clausius-Duhem inequality in a way that does not involve any heat source terms but only internal quantities, by also making use of the energy equation. Namely,
\begin{align}
&-\rho(\dot a_m+\dot\theta s_m)+\sigma:d-\frac{q\cdot\nabla_x\theta}{\theta}\ge0,\label{2nd principe eulerien}\\
& - \Rho\Bigl(\frac{\partial A_m}{\partial t}+S_m\frac{\partial \Theta}{\partial t}\Bigr)+T_{\mathrm{R}}:\nabla_X V-\frac{Q\cdot\nabla_X\Theta}{\Theta}\ge0.\label{2nd principe lagrangien}
  \end{align} 
  Indeed, the conjunction of the above inequality with the energy equation is equivalent to the Clausius-Duhem inequality also with the energy equation, and it turns out to be a convenient expression of the second principle for constitutive purposes. 

  The sum of the first two terms in \eqref{2nd principe eulerien} and \eqref{2nd principe lagrangien} is called the internal dissipation, 
\begin{equation}\label{def dint}
d_{\mathrm{int}}=-\rho(\dot a_m+\dot\theta s_m)+\sigma:d, \quad D_{\text{int}}= - \Rho\Bigl(\frac{\partial A_m}{\partial t}+S_m\frac{\partial \Theta}{\partial t}\Bigr)+T_{\mathrm{R}}:\nabla_X V,
\end{equation} 
in Eulerian and Lagrangian descriptions. 
 Note that $\frac1{\det F}D_{\text{int}}=d_{\text{int}}$ at corresponding points since dissipation is not given as a specific density. The Clausius-Duhem inequality now reads
 \begin{equation}\label{encore une CD}
d_{\mathrm{int}}-\frac{q\cdot\nabla_x\theta}{\theta}\ge0,\quad
D_{\text{int}}-\frac{Q\cdot\nabla_X\Theta}{\Theta}\ge0,
  \end{equation} 
and the Clausius-Planck inequalities now read 
\begin{align}
&d_{\mathrm{int}} \geq 0, \quad q\cdot\nabla_x\theta\le0, \label{Clausius-Planck-eulerienbis}\\
&D_{\text{int}}\geq 0, \quad  Q\cdot\nabla_X\Theta\le0 \label{Clausius-Planck-lagrangienbis}.
  \end{align}

  \subsection{The principle of frame-indifference}\label{principe AIM}
 The principle of frame-indifference is related to the isometries of $E$, see \cite{Truesdell-Noll}, \cite{Gurtin} for a more in-depth discussion. It can be formulated as follows. Given any time-dependent translations
$t\mapsto a(t)\in E$ and rotations $t\mapsto R(t)\in\SO(3)$, we consider two evolutions of the same body for which the correspondence $(x^*,t)=(a(t)+R(t)x, t)$ actually holds between material points. Then the corresponding stresses and heat fluxes should satisfy
$$
\sigma^{*}(x^*, t)=R(t)\sigma(x,t)R(t)^T
$$
and 
$$
q^{*}(x^*,t)=R(t)q(x,t).
$$
In addition, scalar fields must be invariant under the same circumstances. For instance, with similar notation,
$$
a_m^{*}(x^*,t)=a_m(x,t).
$$

In the Lagrangian description, the above situation corresponds to considering two deformations $(X,t)\mapsto \phi(X,t)$ and $(X,t)\mapsto a(t)+R(t)\phi(X,t)$, \emph{i.e.}, the second deformation is a rigid motion superimposed on the first deformation,  with 
\begin{align}
T_{\mathrm{R}}^{*}(X,t)&=R(t)T_{\mathrm{R}}(X,t),\label{AIMTR}\\
Q^{*}(X,t)&=Q(X,t),\label{AIM Q lag}\\
A_m^{*}(X,t)&=A_m(X,t).\label{AIM A lag}
\end{align}

The principle of frame-indifference will be used later on to impose restrictions on acceptable constitutive laws of various kinds.

\section{Thermo-visco-elastic materials with internal variables}\label{le truc principal}
As mentioned before, the equations and inequalities so far apply to all materials and thus cannot be complete. What is needed to describe specific materials are constitutive laws that express the quantities above as functions of thermodynamic variables, thus yielding evolution PDE problems that may have a chance to have solutions, usually under additional hypotheses, and that may be used for numerical simulation purposes. Both aspects are outside of the scope of this article. We will not be overly concerned with regularity issues and assume that all constitutive laws are smooth enough for all computations to be correct.
\subsection{Constitutive assumptions}
The following considerations are mostly in the Lagrangian description, but can also be developed in the Eulerian description.
We are interested in materials the constitutive laws of which depend on an extended set of thermodynamic variables, namely $F$ standing for $\nabla_X\phi(X,t)$, $H$ standing for $\nabla_XV(X,t)$, $\Theta$ standing for $\Theta(X,t)$, and $G$ standing for $\nabla_X\Theta(X,t)$, with in addition two kinds of internal variables, $\Xi\in \R^k$ that eventually intervene in the free energy and $\Pi\in \R^m$ that do not. This is quite a general framework. Considering $H$ as a thermodynamic variable does not seem to be very common, even though it was advocated in \cite{Coleman-Mizel}, see also \cite{Truesdell-Noll}, \cite{Silhavy}, all without internal variables. Such materials can be considered to fall into the category of simple materials defined in \cite{Truesdell-Noll}, albeit not in a thermodynamical context. We refer to \cite{Halphen}, \cite{Holzapfel} or \cite{Lemaitre} for models with internal variables.

We thus assume that we are given functions 
\begin{align*}
&\widehat T_{\mathrm{R}}\colon \Omega\times  \M_3^+\times \M_3 \times\R_+^*\times\R^3\times\R^m\times\R^k\to \M_3,\\
 &\widehat Q\colon \Omega\times \M_3^+\times \M_3\times\R_+^*\times\R^3\times\R^m\times\R^k\to\R^3,\\
 &\widehat S_m\colon \Omega\times  \M_3^+\times \M_3 \times\R_+^*\times\R^3\times\R^m\times\R^k\to \R,\\
&\widehat A_m\colon\Omega\times 
 \M_3^+\times \M_3\times\R_+^*\times\R^3\times\R^m\times\R^k\to\R, 
 \end{align*}
 which serve as constitutive laws for  the first Piol\`a-Kirchhoff stress, the heat flux, the entropy,   and  the free energy respectively, in the sense that
 $$T_{\mathrm{R}}(X,t)=\widehat T_{\mathrm{R}}(X,F(X,t),H(X,t),\Theta(X,t),G(X,t),\Pi(X,t),\Xi(X,t)),$$
 and so on, writing $G(X,t)=\nabla_X\Theta(X,t)$ for brevity. All thermodynamic and internal variables are assumed to enter all constitutive laws in accordance to Truesdell's equipresence principle, \cite{Truesdell-Noll}. To shorten an already cumbersome notation, we will from now on drop the dependence on $X$, which is there to account for possible inhomogeneity and does not play much of a role in the sequel. 
 
 Internal variables are assumed to represent other physical processes that might be present. They are not necessarily observable. They can be scalar-valued, vector-valued or tensor-valued, we regroup all these possibilities within a generic $\R^m$ or $\R^k$.
 The two kinds of internal variables we consider are distinguished according to whether it is possible to assign their value and that of their time derivative independently ($\Pi$) or not $(\Xi)$. In practice, this means that we assume an ordinary differential equation for $\Xi$ of the form 
\begin{equation}\label{l'edo}
\frac{\partial\Xi}{\partial t}(X,t)=\widehat K(F(X,t),H(X,t),\Theta(X,t),G(X,t),\Pi(X,t),\Xi(X,t))
\end{equation}
where $\widehat K\colon \M^+_3\times \M_3\times\R_+^*\times\R^3\times \R^m\times\R^k\to\R^k$ is another given function.

We can easily consider internal variables $\Xi$ taking their values in some proper subset $U$ of $\R^k$, in which case equation \eqref{l'edo} must be required to leave $U$ invariant. 
This equation can also be generalized to a differential inclusion. We leave aside for the moment the question of which initial condition should be imposed on $\Xi$ and assume as a rule the Cauchy problem to be well-posed for any given $F, H, \Theta,G$ and $\Pi$.

The other internal variables $\Pi$ typically will be solutions of evolution PDEs that depend on which phenomena they are supposed to model. The exact form of such PDEs is irrelevant in the ensuing analysis and will never be specified.

The Lagrangian description being arbitrary, it is worth noticing that if a material has the above constitutive laws with respect to one reference configuration, it has constitutive laws of the same form with respect to any other reference configuration (the variable $X$ needs to be retained for this though). 

\subsection{The Coleman-Noll procedure}
We will now perform the Coleman-Noll procedure, \cite{ColemanNoll}, to determine the restrictions that the second principle, in conjunction with the first principle and the dynamics equation, imposes on the above  constitutive laws. 
\begin{proposition}\label{restrictions thermo avec variables internes}
The principles of thermodynamics and the dynamics equation imply that
 
 i) the specific free energy density $\widehat A_m$ is only a function of $F$, $\Theta$ and $\Xi$,
 
   ii) the specific entropy density $\widehat S_m$ is only a function of $F$, $\Theta$ and $\Xi$ with
 $$\widehat S_m(F,\Theta,\Xi)=-\frac{\partial \widehat A_m}{\partial \Theta}(F,\Theta,\Xi),$$

 iii) there is a natural additive decomposition of stress into a thermoelastic part $T_{\mathrm{Re}}$ and a dissipative part $T_{\mathrm{Rd}}$, with constitutive laws
 $$\widehat T_{\mathrm{Re}}(F,\Theta,\Xi)=\Rho\frac{\partial\widehat A_m}{\partial F}(F,\Theta,\Xi)\text{ and }
 \widehat T_{\mathrm{Rd}}=\widehat T_{\mathrm{R}}-\widehat T_{\mathrm{Re}}$$
 yielding a constitutive law for the internal dissipation
  \begin{multline}\label{loi dissipation interne}
 \widehat D_{\mathrm{int}}(F,H,\Theta,G,\Pi,\Xi)=\widehat T_{\mathrm{Rd}}(F,H,\Theta,G,\Pi,\Xi):H\\-\Rho\frac{\partial\widehat A_m}{\partial \Xi}(F,\Theta,\Xi)\cdot \widehat K(F,H,\Theta,G,\Pi,\Xi),
 \end{multline}
satisfying the dissipation inequality
 \begin{equation}\label{dissipation ineq}
 \widehat D_{\mathrm{int}}(F,H,\Theta,G,\Pi,\Xi)-\frac{\widehat Q(F,H,\Theta,G,\Pi,\Xi)\cdot G}{\Theta}\ge 0.
 \end{equation}

 Conversely, if the constitutive laws satisfy i), ii) and iii), then the second principle is satisfied.
\end{proposition}

\begin{proof}
The Coleman-Noll procedure consists in testing the Clausius-Duhem inequality with carefully chosen fields $\phi(X,t)$ and $\Theta(X,t)$, plus chosen internal variable evolutions. For this to be acceptable, it must be checked beforehand that such arbitrary, smooth enough evolutions can be solutions of the dynamics equation \eqref{cons mom lag}  and of the energy equation \eqref{eq d'energie ref}, at least in principle. This is indeed obtained by adjusting the source terms, namely the applied body force density for the dynamics equation
$$
B(X,t)=\Rho\Gamma(X,t){-}{\div\nolimits_{\!X}}\widehat  T_{\mathrm{R}}(F(X,t),H(X,t),\Theta(X,t),G(X,t),\Pi(X,t),\Xi(X,t))
$$
 and the heat source
\begin{multline*}
R(X,t)=\Rho\frac{\partial}{\partial t}\bigl(\widehat E_m(F(X,t),H(X,t),\Theta(X,t),G(X,t),\Pi(X,t),\Xi(X,t))\bigr)\\-\widehat  T_{\mathrm{R}}\bigl(F(X,t),H(X,t),\Theta(X,t),G(X,t),\Pi(X,t),\Xi(X,t)\bigr):H(X,t)\\+{\div\nolimits_X}\bigl(\widehat Q(F(X,t),H(X,t),\Theta(X,t),G(X,t),\Pi(X,t),\Xi(X,t))\bigr).
\end{multline*}
for the energy equation, the constitutive law for the internal energy specific density being of course $\widehat E_m=\widehat A_m+\Theta\widehat S_m$.

We know take the Clausius-Duhem inequality in the form \eqref{2nd principe lagrangien}, substitute inside the constitutive laws and apply the chain rule. The resulting inequality would be excruciatingly long if written in full, so we agree that all hatted quantities that appear are taken at their lengthy list of arguments $F(X,t)$, $H(X,t)$, $\Theta(X,t)$, $G(X,t)$, $\Pi(X,t)$, and $\Xi(X,t)$. We thus obtain
 \begin{multline}\label{ineg fond internes}
- \Rho\Bigl(\frac{\partial \widehat A_m}{\partial \Theta}+\widehat S_m\Bigr)\frac{\partial \Theta}{\partial t}+\Bigl(\widehat T_{\mathrm{R}}-\Rho\frac{\partial\widehat A_m}{\partial F}\Bigr):H\\-\Rho\frac{\partial\widehat A_m}{\partial H}:\nabla_X\Gamma-\Rho\frac{\partial\widehat A_m}{\partial G}\cdot\frac{\partial G}{\partial t}\\
-\Rho\frac{\partial\widehat A_m}{\partial \Xi}\cdot\frac{\partial\Xi}{\partial t}-\Rho\frac{\partial\widehat A_m}{\partial \Pi}\cdot\frac{\partial\Pi}{\partial t}-\frac{\widehat Q\cdot G}{\Theta}\ge0.
 \end{multline}

Let us choose some point $X_0\in\Omega$. Let us be given any $F\in \M_3^+$,  $H,M\in \M_3$, $\Theta\in\R_+^*$, $G,L\in\R^3$, $\Pi,\Pi'\in \R^m$ and $\Xi\in\R^k$. We first choose
 $\phi(X,t)=FX+tHX+\frac{t^2}2MX$, which is orientation-preserving for $t$ small enough, $\Theta(X,t)=\Theta e^{\frac{(X-X_0)\cdot(G+tL)}{\Theta}}$ which is strictly positive. By assumption, there is an evolution $\Pi(X,t)$ such that $\Pi(X_0,0)=\Pi$ and $\frac{\partial\Pi}{\partial t}(X_0,0)=\Pi'$. There is also an evolution $\Xi(X,t)$ satisfying the Cauchy problem with initial datum $\Xi(X_0,0)=\Xi$, substituting all the other evolutions in the right-hand side of the ordinary differential equation \eqref{l'edo} for $\Xi$.
 
With these choices, we obtain $F(X_0,0)=F$, $H(X_0,0)=H$, $\nabla_X\Gamma(X_0,0)=M$, $\Theta(X_0,0)=\Theta$, $\frac{\partial\Theta}{\partial t}(X_0,0)=0$, $G(X_0,0)=G$, $\frac{\partial G}{\partial t}(X_0,0)=L$. Letting $\widehat T_{\mathrm{Rd}}=\widehat T_{\mathrm{R}}-\Rho\frac{\partial\widehat A_m}{\partial F}$ and $\widehat D_{\mathrm{int}}$ be defined as in \eqref{loi dissipation interne},
% \begin{multline*}
% \widehat D_{\mathrm{int}}(F,H,\Theta,G,\Pi,\Xi)=\widehat T_{\mathrm{Rd}}(F,H,\Theta,G,\Pi,\Xi):H
% \\-\Rho\frac{\partial\widehat A_m}{\partial \Xi}(F,H,\Theta,G,\Pi,\Xi)\cdot \widehat K(F,H,\Theta,G,\Pi,\Xi),
% \end{multline*}
 inequality \eqref{ineg fond internes} at $(X_0,0)$ then reads
\begin{equation}\label{ineq intermediaire}
\widehat D_{\mathrm{int}}-\Rho\frac{\partial\widehat A_m}{\partial H}:M-\Rho\frac{\partial\widehat A_m}{\partial G}\cdot L-\Rho\frac{\partial\widehat A_m}{\partial \Pi}\cdot \Pi'
-\frac{\widehat Q\cdot G}{\Theta}\ge0,
\end{equation}
with a similar convention that the  hatted quantities take $(F,H,\Theta,G,\Pi,\Xi)$ as arguments.

Since $M\in \M_3$, $L\in\R^3$ and $\Pi'\in\R^m$ are arbitrary, it follows that $\frac{\partial\widehat A_m}{\partial H}=0$, $\frac{\partial\widehat A_m}{\partial G}=0$ and $\frac{\partial\widehat A_m}{\partial \Pi}=0$. Therefore, $\widehat A_m$ depends neither on $H$, nor on $G$, nor on $\Pi$, which is assertion i). Moreover, inequality \eqref{ineq intermediaire} now boils down to \eqref{dissipation ineq}, \emph{i.e.}, assertion iii).
 
We finally take $\phi(X,t)=FX+tHX$,  $\Theta(X,t)=\Theta+\alpha t+(X-X_0)\cdot G$, with $\Theta>0$, $\alpha\in \R$, $t$ small enough and $X$ in a neighborhood of $X_0$, and the internal variables solution of their respective equations with initial data $\Pi$, $\Pi'$ and $\Xi$. With this choice, we obtain $F(X_0,0)=F$, $H(X_0,0)=H$, $\Theta(X_0,0)=\Theta$, $\frac{\partial\Theta}{\partial t}(X_0,0)=\alpha$, $G(X_0,0)=G$. At $(X_0,0)$, inequality \eqref{ineg fond internes} becomes
$$
 -\alpha\Rho\Bigl(\frac{\partial \widehat A_m}{\partial \Theta}+\widehat S_m\Bigr)+\widehat D_{\mathrm{int}}-\frac{\widehat Q\cdot G}{\Theta}\ge0.
$$
 Since $\alpha$ is arbitrary, it follows that
 $$
\widehat S_m(F,H,\Theta,G,\Pi,\Xi)=-\frac{\partial \widehat A_m}{\partial \Theta}(F,\Theta,\Xi),
 $$
which depends neither on $H$, nor on $G$, nor on $\Pi$, that is to say, assertion ii). 

At this point, we note that the function $\widehat D_{\mathrm{int}}$ is actually the constitutive law for the internal dissipation $D_{\mathrm{int}}$ defined earlier in \eqref{def dint}.

Conversely, if all these constitutive assumptions hold, then inequality \eqref{2nd principe lagrangien} is always satisfied. 
\end{proof}

\begin{remark}One of the main outcomes of the Coleman-Noll procedure is that, starting from four assumed constitutive laws, the second principle implies that there are only three master constitutive laws that can be specified independently, namely $\widehat A_m$, $\widehat Q$ and $ \widehat T_{\mathrm{Rd}}$. The constitutive laws for entropy and thermoelastic stress are in a sense included in the free energy constitutive law. 

As far as internal variables are concerned, some of them ($\Xi$) require a fifth constitutive law $\widehat K$, which we will call a flow rule as it is called in certain contexts, as right-hand side of their evolution ordinary differential equation. The second principle makes this flow rule appear in the internal dissipation and consequently in the dissipation inequality \eqref{dissipation ineq}. The other internal variables ($\Pi$) are hardly constrained by the second principle at all, except inasmuch as they enter as arguments in the dissipation inequality.
\end{remark}

\begin{remark}Since  inequality \eqref{2nd principe lagrangien}  is equivalent to the second principle in the presence of the first principle, we thus have a set of necessary and sufficient constitutive assumptions for this family of thermo-visco-elastic materials to satisfy the second principle.

The dynamics equation now assumes the form
\begin{align*}
&\Rho\Gamma-{\div\nolimits_X}\widehat T_{\mathrm{R}}(F,H,\Theta,G,\Pi,\Xi)=B,\\
&F\widehat T_{\mathrm{R}}(F,H,\Theta,G,\Pi,\Xi)^T=\widehat T_{\mathrm{R}}(F,H,\Theta,G,\Pi,\Xi)F^T,
\end{align*}
in $\Omega\times I$, where $I$ is some time interval. The second relation, which reflects the symmetry of the Cauchy stress, is more of a constitutive assumption than an actual equation. The first equation has unknowns $\phi$, $\Theta$, $\Pi$ and $\Xi$. It is of second order in time with respect to $\phi$, thus should be complemented with initial conditions for $\phi$ and $V$. Additionally, boundary values and initial conditions should be provided for all unknowns. 

Let us remark that the second principle says nothing about the well-posedness of this system. Indeed, in the elastic case ($ \widehat T_{\mathrm{Rd}}=0$, $ \widehat A_m(F,\Theta)=\widehat W_m(F)+\widehat V_m(\Theta)$, no internal variables), the mechanical part of the Clausius-Planck inequality is automatically satisfied with zero internal dissipation, irrespective of whether or not the dynamics equations are hyperbolic. 

The energy equation can be rewritten as a heat equation
$$
-\Rho\Theta\frac{\partial^2 \widehat A_m}{\partial \Theta^2}\frac{\partial\Theta}{\partial t}+{\div\nolimits_X}\widehat Q=\widehat D_{\mathrm{int}}+\Rho\Theta\frac{\partial\widehat T_{\mathrm{Re}}}{\partial\Theta}:\nabla_XV+\Rho\Theta\frac{\partial^2\widehat A_m}{\partial\Theta\partial\Xi}\cdot\widehat K+R.
$$
Again, there is no provision in the second principle for this heat equation to be parabolic. It can be ill-posed like a backward heat equation, but still satisfy the second principle.
Note that the internal dissipation plays the role of a heat source, hence its name.

Finally, these equations are coupled to the evolution equations for both kinds of  internal variables $\Pi$ and $\Xi$.
\end{remark}
\begin{remark}\label{Fourier classique}
The classical Fourier law assumes the Lagrangian form 
\begin{equation}\label{classique Fourier loi}
\widehat Q(F,G)=-k(F^TF)^{-1}G\text{ with } k>0,
\end{equation} which corresponds to  $\widehat q(g)=-kg$ in the Eulerian description (so that ${\div\nolimits_x}\widehat q(\nabla_x\theta)=-k\Delta_x \theta$). 
Since $\widehat q(g)\cdot g=-k\|g\|^2$, or  $\widehat Q(F,G)\cdot G=-kG^T(F^TF)^{-T}G=-k\|F^{-T}G\|^2$, this law satisfies the thermal part of the Clausius-Planck inequalities~\eqref{Clausius-Planck-eulerienbis}-\eqref{Clausius-Planck-lagrangienbis}. 
\end{remark}

\begin{remark}\label{irreversibilite}It was claimed earlier that if the Clausius-Duhem inequality is strict, then the evolution is irreversible. Now the starting point of the Coleman-Noll procedure is to be able to consider arbitrary evolutions and in the end, the ensuing restrictions on the constitutive laws make them satisfy the second principle for all evolutions, including the time-reversed version of any given evolution.  However, the source terms in the Coleman-Noll procedure are not the same ones as in the previous context, but are drastically modified. In particular, the heat sources and fluxes are not the opposite time-reversed heat sources and fluxes of the given evolution.

This is easier seen on the Clausius-Planck version. Indeed, since $G$ is invariant under time reversal, if the constitutive law for the heat flux constitutive law is such that $\widehat Q\cdot G<0$ for $G\neq 0$, then it is not possible to reverse the direction of the heat flux. Hence, the heat source for the time-reversed evolution cannot be the opposite of the time-reversed heat source in order to occur as an entropy source or in the energy equation. Similarly, taking for instance $ \widehat T_{\mathrm{Rd}}(F,H,\Theta)=H$ (not a good choice as will be seen a little later on), then the internal dissipation is strictly positive for $H\neq 0$, but time reversal changes $H$ into $-H$ and thus $ \div_X\widehat T_{\mathrm{Rd}}$ into $-\div_X\widehat T_{\mathrm{Rd}}$, and the body forces must be modified accordingly in the dynamics equation, which is not just by time reversal. 
\end{remark}

We can see that for our thermo-visco-elastic materials, there is no reason in general for the Clausius-Planck inequalities to be equivalent to the Clausius-Duheim inequality. They are however partly or totally implied by the Clausius-Duhem inequality in certain cases.

\begin{proposition}\label{CP OK}Assume that $\widehat T_{\mathrm{Rd}}$ and $\widehat K$ do not depend on $G$. Then the Clausius-Duhem inequality \eqref{encore une CD} implies the mechanical part of the Clausius-Planck inequalities \eqref{Clausius-Planck-eulerienbis}-\eqref{Clausius-Planck-lagrangienbis}.   If furthermore $\widehat Q$ does not depend on $H$ nor on $\Xi$, then 
\begin{equation}\label{Clausius-Planck du pauvre}
\widehat Q(F,\Theta,G,\Pi)\cdot G\le \Theta\inf_{(H,\Xi)\in \M_3\times\R^k}\widehat D_{\mathrm{int}}(F,H,\Theta,\Pi,\Xi).
\end{equation}

If, in addition, we assume i) that for all $F\in\M_3^+$ and $\Theta\in\R_+^*$, there exists $\Xi\in\R^k$ such that $\frac{\partial\widehat A_m}{\partial\Xi}(F,\Theta,\Xi)=0$ or ii) that there is no internal variable $\Xi$ at all, then the thermal part of the Clausius-Planck inequalities is also implied by the Clausius-Duhem inequality.
\end{proposition}

\begin{proof}
Indeed, as $\widehat T_{\mathrm{Rd}}$ and $\widehat K$ do not depend on $G$, the constitutive law  $\widehat D_{\mathrm{int}}$ defined by \eqref{loi dissipation interne} does not depend on $G$ either, therefore the dissipation inequality reads
$$
 \widehat D_{\mathrm{int}}(F,H,\Theta,\Pi,\Xi)-\frac{\widehat Q(F,H,\Theta,G,\Pi,\Xi)\cdot G}{\Theta}\ge 0.
 $$
Taking $G=0$, we observe that the mechanical part of the Clausius-Planck inequalities ensues.

Now if $\widehat Q$ does not depend on $(H,\Xi)$, then \eqref{Clausius-Planck du pauvre} obviously holds. Of course, the right-hand side is then nonnegative. 

Finally, in case i), we take $\Xi\in\R^k$ such that $\frac{\partial\widehat A_m}{\partial\Xi}(F,\Theta,\Xi)=0$ and $H=0$ and in case ii), just $H=0$, we see that the right-hand side of \eqref{Clausius-Planck du pauvre} vanishes, which implies the thermal part of the Clausius-Planck inequalities.
 \end{proof}
 
\begin{remark}Proposition \ref{CP OK} thus provides us with two rather general situations in which the Clausius-Duhem and Clausius-Planck inequalities are equivalent, so that nothing is lost by working with the Clausius-Planck inequalities. From now on and for simplicity, we will assume that the internal dissipation does not depend on $G$, that the heat flux does not depend on $(H,\Xi)$, and that we are in cases i) or ii) above.
\end{remark}

\subsection{Dissipation potentials}
A natural question is how can one ensure that, given constitutive laws for the free energy and the heat flux, the Clausius-Planck inequalities \eqref{Clausius-Planck-lagrangienbis} hold, \emph{i.e.}, how to find an appropriate constitutive law for the dissipative part of the stress and appropriate flow rule for the internal variables $\Xi$ in this respect. 

The thermal part of  the Clausius-Planck inequalities is a requirement on $\widehat Q$ itself, namely that 
$$
\widehat Q(F,\Theta,G,\Pi)\cdot G\le 0,
$$
for all $F$, $\Theta$, $G$, and $\Pi$. We have seen in remark \ref{Fourier classique} that it is satisfied by the classical Fourier law, and we will see later on generalizations thereof that also satisfy it.

The mechanical/internal variable part of the Clausius-Planck inequalities, 
$$
\widehat D_{\mathrm{int}}(F,H,\Theta,\Pi,\Xi)\ge 0,
$$ 
 is easily satisfied via a dissipation potential. Dissipation potentials are often introduced in a fairly obscure fashion, so let us just stress here that they have no a priori physical significance, even though the choice of a specific dissipation potential should reflect some modeling concern. 

\begin{proposition}\label{potentiel de dissipation cas visqueux et var internes}
Let $\widehat P_{\mathrm{diss}}\colon
 \M_3^+\times \M_3\times\R_+^*\times\R^m\times\R^k\to\R_+$ be a function that is convex with respect to $(H,\Lambda)$, where $\Lambda\in \R^k$ is the last variable, and such that $\widehat P_{\mathrm{diss}}(F,0,\Theta,\Pi,0)=0$. Then the constitutive laws 
 $$\widehat T_{\mathrm{Rd}}(F,H,\Theta,\Pi,\Xi)=\frac{\partial \widehat P_{\mathrm{diss}}}{\partial H}\Bigl(F,H,\Theta,\Pi,\frac{\partial\widehat A_m}{\partial\Xi}(F,\Theta,\Xi)\Bigr)$$
and
  $$\widehat K(F,H,\Theta,\Pi,\Xi)=-\frac1{\Rho}\frac{\partial  \widehat P_{\mathrm{diss}}}{\partial\Lambda}\Bigl(F,H,\Theta,\Pi,\frac{\partial\widehat A_m}{\partial\Xi}(F,\Theta,\Xi)\Bigr)$$
are such that $\widehat D_{\mathrm{int}}(F,H,\Theta,\Pi,\Xi)\ge 0$ for all values of its arguments. 
\end{proposition}

\begin{proof}This is due to the fact that a positive, differentiable, convex function $J$ with $J(0)=0$ is such that  $\langle dJ(u),u\rangle\ge 0$. Here we consider $F$, $\Theta$ and $\Pi$ as parameters, $u=(H,\Lambda)$ and $J(u)=\widehat P_{\mathrm{diss}}(F,H,\Theta,\Pi,\Lambda)$, from which it follows that 
$$
\frac{\partial \widehat P_{\mathrm{diss}}}{\partial H}(F,H,\Theta,\Pi,\Lambda):H+\frac{\partial \widehat P_{\mathrm{diss}}}{\partial \Lambda}(F,H,\Theta,\Pi,\Lambda)\cdot\Lambda\ge 0,
$$
and it suffices to take $\Lambda=\frac{\partial\widehat A_m}{\partial\Xi}(F,\Theta,\Xi)$ to obtain the desired inequality.
\end{proof} 

So dissipation potentials are more like recipes to construct dissipative constitutive laws, given the fact that it is easier to check that a function is nonnegative, convex and $0$ at the origin than to construct $\widehat T_{\mathrm{Rd}}$ and $\widehat K$ from scratch. Note that we could also introduce a similar notion of diffusion potential to construct admissible heat fluxes, but this seems less common. In the case when the Clausius-Duhem inequality does not reduce to the Clausius-Planck inequalities, a combined dissipation-diffusion potential will also work.

\begin{remark}\label{example 1}We will later on describe more elaborate exemples with internal variables, but let us for now just take $\widehat A_m(F,\Theta)=\widehat V_m(\Theta)$ so that $\widehat T_{\mathrm{Re}}=0$, and $\widehat Q(F,G)=-kC^{-1}G$ with $C=F^TF$ and $k>0$, the classical Fourier law. We consider a purely kinematical viscosity  $\sigma = 2 \nu d$, $\nu>0$, that corresponds to $\widehat T_{\mathrm{Rd}}(F,H)=\nu\det F\bigl(HC^{-1}+F^{-T}H^TF^{-T}\bigr)$. It can be checked after some computation that this stress tensor  derives from the dissipation potential $\widehat P_{\mathrm{diss}}(F,H)=\frac{\nu\det F}2\|HC^{-1}+F^{-T}H^TF^{-T}\|^2$, which is clearly nonnegative, convex,  and equal to $0$ when $H=0$.  The dynamics equation is in this case decoupled from the heat equation, as it does not involve the temperature. The heat equation however still has the  internal dissipation, which is in general strictly positive, as a heat source term.

If we add a term of the form $W_m(F)$ to the free energy, we obtain a stress tensor which is the sum of a purely elastic term and of a purely kinematically viscous term, and so on and so forth. 

Of course, if there are internal variables $\Xi$, they will be responsible for another kind of dissipation, for instance viscosity effects that are not kinematical, see Section \ref{modeles maxwell 3d}. 
\end{remark}

%Let us return briefly to the question of the Clausius-Planck inequalities vs. the Clausius-Duhem inequality in the present context.
%
%\begin{proposition}\label{finalement CP pas si mal}Assume that for all $F\in\M_3^+$ and $\Theta\in\R_+^*$, there exists $\Xi\in\R^k$ such that $\frac{\partial\widehat A_m}{\partial\Xi}(F,\Theta,\Xi)=0$. Then the thermal part of the Clausius-Planck inequalities is satisfied.
%\end{proposition}
%
%\begin{proof}In the notation of the previous proof, by definition of a dissipation potential, $dJ((0,0))=0$, which says that $\widehat D_{\mathrm{int}}(F,0,\Theta,\Pi,0)= 0$ and inequality \eqref{Clausius-Planck du pauvre} then implies that $
%\widehat Q(F,\Theta,G,\Pi)\cdot G\le 0
%$.
%\end{proof}

\section{Frame-indifference and symmetries for thermo-visco-elastic materials}\label{sursection AIM visqueux}
The principle of frame-indifference applies to the materials under consideration. Since the nature of the internal variables is left unspecified at this point, they cannot be taken into account in a generic manner in this respect. They have to be dealt with on a case-by-case basis. We thus just consider here thermo-visco-elastic materials without internal variables and will return later to specific examples with internal variables. 

In all the proofs of this section, $\Theta$ only plays the role of a parameter, so we will systematically omit it for brevity.

\subsection{Frame-indifference}\label{section AIM visqueux}

The following result expresses frame-indifference for the stresses in the Lagrangian description. A similar result, however not expressed in the same way, can be found in \cite{Silhavy}. We use the notation  $\Sym(M)$ for the symmetric part of a matrix $M$.

\begin{proposition}\label{AIM lagrangien visqueux}The constitutive law for the first Piol\`a-Kirchhoff stress tensor is compatible with the principle of frame-indifference if and only if, for all $F\in \M_3^+$, $H\in \M_3$, $\Theta\in \R_+^*$ and $R\in \SO(3)$, we have
\begin{equation}\label{AIM visqueux 1}
\widehat T_{\mathrm{R}}(RF,RH,\Theta)=R\widehat T_{\mathrm{R}}(F,H,\Theta),
\end{equation}
and
\begin{equation}\label{AIM visqueux 2}
\widehat T_{\mathrm{R}}(F,H,\Theta)=\widehat T_{\mathrm{R}}\bigl(F, \Sym(H F^{-1})F,\Theta\bigr).
\end{equation}
\end{proposition}
\begin{proof}Let us be given an arbitrary deformation $\phi$ and an arbitrary  rigid motion $t\mapsto (a(t),R(t))$. Now let $\phi^*(X,t)=a(t)+R(t)\phi(X,t)$, $F^*(X,t)=\nabla_X\phi^*(X,t)$, $V^*(X,t)=\frac{\partial \phi^*}{\partial t}(X,t)$ and $H^*(X,t)=\nabla_XV^*(X,t)$. Clearly,
$$F^*(X,t)=R(t)F(X,t)\text{ and }H^*(X,t)=R(t)H(X,t)+R'(t)F(X,t).$$

The principle of frame-indifference requires that the relation $T_{\mathrm{R}}^{*}=R(t)T_{\mathrm{R}}$ should always be observed, that is to say in terms of the constitutive law, that the latter should satisfy
\begin{equation}\label{AIM loi de comp}
\widehat T_{\mathrm{R}}(F^*(X,t),H^*(X,t))=R(t)\widehat T_{\mathrm{R}}(F(X,t),H(X,t)).
\end{equation}
We first take $R(t)=R$ a constant rotation, $a(t)=0$ and $\phi(X,t)=FX+tHX$ with $F\in\M_3^+$ and $H\in\M_3$ arbitrary, $t$ small enough. Then, \eqref{AIM visqueux 1} follows from writing \eqref{AIM loi de comp} at $t=0$.

Let us now keep the same deformation $\phi$, still $a(t)=0$, but  $R(t)=e^{tW}$ with $W$ skew-symmetric and arbitrary. We have $H^*(X,0)=H+WF$ and $R(0)=I$. Then, \eqref{AIM loi de comp} implies that
$$\widehat T_{\mathrm{R}}(F,H)=\widehat T_{\mathrm{R}}(F,H+WF)=\widehat T_{\mathrm{R}}(F,(HF^{-1}+W)F).$$
Taking for $W$ the opposite of the skew-symmetric part of $HF^{-1}$, we obtain \eqref{AIM visqueux 2}.

Conversely, we now assume \eqref{AIM visqueux 1} and \eqref{AIM visqueux 2}, and consider arbitrary deformations and rigid motions. 
The translational invariance is automatic here, since $a(t)$ does not appear in the constitutive law. With the same notation as above, but without writing the $X$ and $t$ arguments for brevity, we see that
$$H^*(F^*)^{-1}=RH(F^*)^{-1}+R'F(F^*)^{-1}=RHF^{-1}R^T+R'R^T.$$
Taking the symmetric part, we obtain
$$\Sym(H^*(F^*)^{-1})=R\Sym(HF^{-1})R^T$$
because $R'R^T$ is skew-symmetric. Therefore,
$$\Sym(H^*(F^*)^{-1}) F^*=R\Sym(HF^{-1})F,$$
so that by \eqref{AIM visqueux 1} and \eqref{AIM visqueux 2}, 
\begin{align*}
T_{\mathrm{R}}^{*}&=\widehat T_{\mathrm{R}}(F^*,H^*) =\widehat T_{\mathrm{R}} (F^*,\Sym(H^*(F^*)^{-1}) F^*)\\
&=\widehat T_{\mathrm{R}}(RF,R\Sym(HF^{-1})F) =R\widehat T_{\mathrm{R}}(F,\Sym(HF^{-1})F)\\
&=R\widehat T_{\mathrm{R}}(F,H)=RT_{\mathrm{R}}.
\end{align*}
Therefore, the constitutive law is compatible with the principle of frame-indifference.
\end{proof}

Observe that the example of remark \ref{example 1} satisfies \eqref{AIM visqueux 1} and \eqref{AIM visqueux 2} and is thus frame-indifferent, while that of remark \ref{irreversibilite} is not.

In terms of the Cauchy stress tensor, the previous results can be rewritten in the following way.
\begin{proposition}\label{AIM eulerien visqueux}There exists a function $\check\sigma\colon\M_3^+\times\Sym_3\times\R^*_+\to\Sym_3$ such that the Cauchy stress tensor is given by
$$\sigma(\phi(X,t),t)=\check\sigma\bigl(F(X,t),\Sym(H(X,t)F^{-1}(X,t)),\Theta(X,t)\bigr),$$
with
$$
\check\sigma(RF,RMR^T,\Theta)=R\check\sigma(F,M,\Theta)R^T
$$
for all $F\in \M_3^+$, $M\in\Sym_3$, $R\in \SO(3)$ and $\Theta\in \R_+^*$.
In the Eulerian description, this reads
$$\sigma(x,t)=\check\sigma\bigl(\phi^{-1}(x,t),\nabla_x\phi^{-1}(x,t),d(x,t),\theta(x,t)\bigr),$$
where $d=\Sym(\nabla_xv)$ is the stretching  tensor, with 
$$\check\sigma(X,fR^T,RdR^T,\theta)=R\check\sigma(X,f,d,\theta)R^T$$
for all possible arguments.
\end{proposition}

%Scalar constitutive laws must also obey the principle of frame-indifference, \emph{i.e.}, yield invariant quantities. For instance, the frame-indifference of the first Piolà-Kirchhoff stress implies the frame-indifference of the internal dissipation. Frame-indifference also applies to the internal energy $\widehat A_m$ and to the dissipation potential $\widehat P_{\mathrm{diss}}$ as well if assumed, even though the latter is not necessarily a physically relevant quantity. We use the notation $\Sym^+_3$ for the set of $3\times3$ symmetric positive definite matrices.

Scalar constitutive laws must also obey the principle of frame-indifference, \emph{i.e.}, yield invariant quantities, see equation \eqref{AIM A lag}. %For instance, the internal dissipation is frame-indifferent, which also follows from the frame-indifference of the first Piolà-Kirchhoff stress. Frame-indifference also applies to the internal energy $A_m$. 
Dissipation potentials $\widehat P_{\mathrm{diss}}$ however are constitutive laws that do not necessarily correspond to physically relevant quantities. We are nonetheless at liberty to assume that they are compatible with frame-indifference as well, as though they were used to compute physical quantities, see Proposition \ref{ca marche quand meme}. We use the notation $\Sym^+_3$ for the set of $3\times3$ symmetric positive definite matrices.

\begin{proposition}\label{AIM potentiel de dissipation}The free energy constitutive law $\widehat A_m$ is compatible with the principle of frame-indifference if and only if for all $F\in \M_3^+$, $\Theta\in \R_+^*$ and $R\in \SO(3)$
\begin{equation}\label{AIM energie libre}
\widehat A_m(RF,\Theta)=\widehat A_m(F,\Theta)
\end{equation}
in which case there exists a function $\widetilde A_m\colon\Sym^+_3\times\R_+^*\to\R$ such that
\begin{equation}\label{AIM energie libre bis}
\widehat A_m(F,\Theta)=\widetilde A_m(F^TF,\Theta).
\end{equation}
A dissipation potential constitutive law $\widehat P_{\mathrm{diss}}$ is compatible with the principle of frame-indifference if and only if for all $F\in \M_3^+$, $H\in \M_3$, $\Theta\in \R_+^*$ and $R\in \SO(3)$
\begin{equation}\label{AIM dissipation 1}
\widehat P_{\mathrm{diss}}(RF,RH,\Theta)=\widehat P_{\mathrm{diss}}(F,H,\Theta),
\end{equation}
and
\begin{equation}\label{AIM dissipation 2}
\widehat P_{\mathrm{diss}}(F,H,\Theta)=\widehat P_{\mathrm{diss}}\bigl(F,\Sym(HF^{-1})F,\Theta\bigr).
\end{equation}
\end{proposition}
\begin{proof}The proof of \eqref{AIM energie libre}-\eqref{AIM energie libre bis} is classical in nonlinear elasticity. The proof of \eqref{AIM dissipation 1}-\eqref{AIM dissipation 2} follows the lines of the proof of Proposition \ref{AIM lagrangien visqueux}.
\end{proof}

There does not seem to be a nice representation formula such as \eqref{AIM energie libre bis} for a frame-indifferent dissipation potential.

In addition to obtaining frame-indifferent constitutive laws for the first Piolà-Kirchhoff stress, there is a bonus to assuming frame-indifference for the dissipation potential.  

\begin{proposition}\label{invariance bonus inattendue}If the dissipation potential $\widehat P_{\mathrm{diss}}$ is frame-indifferent, so is $\widehat T_{\mathrm{R}}$. In addition, the symmetry condition $F\widehat T_{\mathrm{R}}^T=\widehat T_{\mathrm{R}}F^T$ is  automatically satisfied.
\end{proposition}

\begin{proof}By definition of the thermoelastic and dissipative parts of the stress, we have
$$\widehat T_{\mathrm{R}}(F,H)=\Rho\frac{\partial \widehat A_m}{\partial F}(F)+\frac{\partial \widehat P_{\mathrm{diss}}}{\partial H}(F,H).$$
Since $\widehat A_m$ is frame-indifferent, so is $\widehat T_{\mathrm{Re}}(F)=\Rho\frac{\partial \widehat A_m}{\partial F}(F)$, as is well known. Morevover,  $F\widehat T_{\mathrm{Re}}(F)^T=\widehat T_{\mathrm{Re}}(F)F^T$ because of \eqref{AIM energie libre bis}. 

We just need to consider the dissipative part  $\widehat T_{\mathrm{\mathrm{Rd}}}(F,H)=\frac{\partial \widehat P_{\mathrm{diss}}}{\partial H}(F,H)$. Differentiating \eqref{AIM dissipation 1} with respect to $H$, we obtain
$$\widehat T_{\mathrm{Rd}}(RF,RH)=R\widehat T_{\mathrm{Rd}}(F,H),$$
for all $F$, $H$ and $R$, \emph{i.e.}, \eqref{AIM visqueux 1}.

Setting now $M=\Sym(HF^{-1})$, we have $H=MF+WF$ with $W$ skew-symmetric. Relation \eqref{AIM dissipation 2} may be rewritten as
\begin{equation}\label{invariance en passant}
\widehat P_{\mathrm{diss}}(F,MF+WF)=\widehat P_{\mathrm{diss}}(F,MF)
\end{equation}
for all symmetric $M$ and skew-symmetric $W$.
Differentiating relation \eqref{invariance en passant} with respect to $W$, we obtain
$$\frac{\partial\widehat P_{\mathrm{diss}}}{\partial H}(F,H)F^T:Z=0$$
for all skew-symmetric $Z$, so that $\widehat T_{\mathrm{Rd}}(F,H)F^T$ is symmetric or $F\widehat T_{\mathrm{Rd}}(F)^T=\widehat T_{\mathrm{Rd}}(F)F^T$. Adding the elastic and dissipative parts, we obtain the symmetry $F\widehat T_{\mathrm{R}}^T=\widehat T_{\mathrm{R}}F^T$.

We next differentiate \eqref{invariance en passant} with respect to $M$, and obtain that
$$\Bigl(\frac{\partial\widehat P_{\mathrm{diss}}}{\partial H}(F,MF+WF)F^T-\frac{\partial\widehat P_{\mathrm{diss}}}{\partial H}(F,MF)F^T\Bigr) :N=0$$
for all symmetric $N$. The term between parentheses is a difference of symmetric matrices, it therefore vanishes, so that 
$$\widehat T_{\mathrm{Rd}}(F,H)=\widehat T_{\mathrm{Rd}}(F,\Sym(HF^{-1})F),$$
that is to say \eqref{AIM visqueux 2}.
\end{proof}

\begin{remark}It must be noticed that without the assumption of a frame-indifferent dissipation potential, the symmetry of the Cauchy stress is a constitutive restriction that must be independently imposed on the constitutive law $\widehat T_{\mathrm{Rd}}$, since the thermoelastic part $\widehat T_{\mathrm{Re}}$ of the stress already has the required symmetry by the frame-indifference of $\widehat A_m$. 
%In hyperelasticity, it is well known that if the first Piolà-Kirchhoff stress constitutive law is frame-indifferent, then so is the stored energy function, which in turns implies that the Cauchy stress constitutive law is automatically symmetric. There is no such thing for dissipative stresses deriving from a dissipation potential as the following counter-example shows.
%
%Let $W_0$ be a nonzero, skew-symmetric matrix. We take  $\widehat P_{\mathrm{diss}}(F,H)=(FW_0):H$. With this choice, it follows that $\widehat T_{\mathrm{Rd}}(F,H)=FW_0$ is  a nonzero constitutive law for the dissipative stress which satisfies \eqref{AIM visqueux 1} and \eqref{AIM visqueux 2}, hence is frame-indifferent. However, 
%$$\widehat T_{\mathrm{Rd}}(F,H)F^T=FW_0F^T=-FW_0^TF^T=-F\widehat T_{\mathrm{Rd}}(F,H)^T
%$$
%so that the Cauchy stress tensor is nonzero and skew-symmetric, in particular, it is not symmetric. Consequently, this dissipation potential is not frame-indifferent, nor can it be replaced by another frame-indifferent potential, by proposition \ref{invariance bonus inattendue}. Note that this counter-example does not satisfy the mechanical part of the Clausius-Planck inequalities.
\end{remark}

We have however the rather remarkable following result.
\begin{proposition}\label{incroyable}
Let us assume that $\widehat T_{\mathrm{Rd}}$ is frame-indifferent and derives from a dissipation potential $\widehat P_{\mathrm{diss}}$. In addition, we assume that the mechanical part of the Clausius-Planck inequality is satisfied. Then
$F\widehat T_{\mathrm{Rd}}^T=\widehat T_{\mathrm{Rd}}F^T$.
\end{proposition}

\begin{proof}Let us introduce some convenient notation. We will use the variable $L=HF^{-1}\in\M_3$ together with $F\in\M_3^+$ and define
$$
\widehat \Sigma(F,L)=\widehat T_{\mathrm{Rd}}(F,LF)F^T,\quad
\check P_{\mathrm{diss}}(F,L)=\widehat P_{\mathrm{diss}}(F,LF).
$$
With this notation, it is fairly clear that  \eqref{AIM visqueux 2} is equivalent to 
\begin{equation}\label{une symetrie de plus}
\widehat \Sigma(F,L)=\widehat \Sigma(F,\Sym(L))
\end{equation}
and that 
$$
\widehat \Sigma(F,L)=\frac{\partial\check P_{\mathrm{diss}} }{\partial L}(F,L),
$$
which is thus still a gradient. In particular, for all indices $i,j,k$, and $l$, we have 
$$
\frac{\partial \widehat \Sigma_{ij}}{\partial L_{kl}}=\frac{\partial \widehat \Sigma_{kl}}{\partial L_{ij}}.
$$

We first show that $\Skew(\widehat\Sigma)$ does not depend on $L$, so is a function of $F$ only. For this, we differentiate \eqref{une symetrie de plus} with respect to $L$, which yields (without writing the $(F,L)$ arguments)
$$
\frac{\partial \widehat \Sigma_{ij}}{\partial L_{kl}}=\frac12\Bigl(\frac{\partial \widehat \Sigma_{ij}}{\partial L_{kl}}+\frac{\partial \widehat \Sigma_{ij}}{\partial L_{lk}}\Bigr)=\frac{\partial \widehat \Sigma_{ij}}{\partial L_{lk}}.
$$
Using the gradient remark above, we see that in fact
$$
\frac{\partial \widehat \Sigma_{ij}}{\partial L_{kl}}=\frac12\Bigl(\frac{\partial \widehat \Sigma_{kl}}{\partial L_{ij}}+\frac{\partial \widehat \Sigma_{lk}}{\partial L_{ij}}\Bigr).
$$
Therefore
\begin{align*}
\frac{\partial \bigl(\widehat \Sigma_{ij}- \widehat \Sigma_{ji}\bigr)}{\partial L_{kl}}&=\frac12\Bigl(\frac{\partial \widehat \Sigma_{kl}}{\partial L_{ij}}+\frac{\partial \widehat \Sigma_{lk}}{\partial L_{ij}}
-\frac{\partial \widehat \Sigma_{kl}}{\partial L_{ji}}-\frac{\partial \widehat \Sigma_{lk}}{\partial L_{ji}}\Bigr)\\
&=\frac12\Bigl(\frac{\partial \widehat \Sigma_{kl}}{\partial L_{ij}}-\frac{\partial \widehat \Sigma_{kl}}{\partial L_{ji}}+\frac{\partial \widehat \Sigma_{lk}}{\partial L_{ij}}
-\frac{\partial \widehat \Sigma_{lk}}{\partial L_{ji}}\Bigr)=0.
\end{align*}
We thus have shown that 
$$\Skew(\widehat\Sigma)(F,L)=\widehat W(F)$$
for some skew-symmetric valued function $\widehat W$.

Going back to the original variables, it follows that 
$$
\widehat T_{\mathrm{Rd}}(F,H)=\Sym\bigl(\widehat \Sigma(F,L)\bigr)F^{-T}+\widehat W(F)F^{-T}.
$$
The mechanical part of the Clausius-Planck inequalities then reads
$$
0\le \widehat T_{\mathrm{Rd}}(F,H):H=\Sym\bigl(\widehat \Sigma(F,L)\bigr):(HF^{-1})+\widehat W(F):(HF^{-1}).
$$
In particular, we can choose $HF^{-1}$ to be any skew-symmetric matrix $Z$, so that 
$$0\le \widehat W(F):Z\text{ and thus }\widehat W(F)=0.$$
This is exactly saying that $\widehat T_{\mathrm{Rd}}(F,H)F^T$ is symmetric.
\end{proof}

\begin{remark}Quite surprisingly, the Clausius-Planck inequalities are crucial here. Indeed,
let $W_0$ be a nonzero, skew-symmetric matrix. We take  $\widehat P_{\mathrm{diss}}(F,H)=(FW_0):H$. With this choice, it follows that $\widehat T_{\mathrm{Rd}}(F,H)=FW_0$ is  a nonzero constitutive law for the dissipative stress which satisfies \eqref{AIM visqueux 1} and \eqref{AIM visqueux 2}, hence is frame-indifferent. However, 
$$\widehat T_{\mathrm{Rd}}(F,H)F^T=FW_0F^T=-FW_0^TF^T=-F\widehat T_{\mathrm{Rd}}(F,H)^T
$$
so that the Cauchy stress tensor is nonzero and skew-symmetric, in particular, it is not symmetric. Consequently, this dissipation potential is not frame-indifferent, nor can it be replaced by another frame-indifferent potential, by proposition \ref{invariance bonus inattendue}. Naturally, this counter-example does not satisfy the mechanical part of the Clausius-Planck inequalities.
\end{remark}

%Obviously, the frame-indifference of the dissipative part of the stress does not imply the frame-indifference of all corresponding dissipation potentials, as opposed to what happens for instance in nonlinear elasticity. Indeed, if we add to a given frame-indifferent dissipation potential yielding a frame-indifferent dissipative stress any non frame-indifferent function of $F$, such as $F_{11}$, we obtain a non frame-indifferent dissipation potential that still yields the same frame-indifferent dissipative stress. Nonetheless,

We now are in a position to nicely round up the questions of frame-indifference of the dissipation potential and symmetry of the Cauchy stress.

\begin{proposition}\label{ca marche quand meme}Let us assume that $\widehat T_{\mathrm{Rd}}$ is frame-indifferent and derives from a dissipation potential $\widehat P_{\mathrm{diss}}$ that is nonnegative, $0$ at $H=0$ and convex with respect to $H$. Then, the potential
$$
\overline P_{\mathrm{diss}}(F,H,\Theta)=\int_{\SO(3)}\widehat P_{\mathrm{diss}}(RF,RH,\Theta)\,d\mu,
$$
where $\mu$ is the left Haar measure on $\SO(3)$, is nonnegative, $0$ at $H=0$, convex with respect to $H$, yields the same dissipative stress $\widehat T_{\mathrm{Rd}}$ as $\widehat P_{\mathrm{diss}}$ and is frame-indifferent.
\end{proposition}

\begin{proof}
It is clear that $\overline P_{\mathrm{diss}}$ is convex with respect to $H$, nonnegative and that $\overline P_{\mathrm{diss}}(F,0)=0$. Moreover,
\begin{multline*}
\frac{\partial\overline P_{\mathrm{diss}}}{\partial H}(F,H)=\int_{\SO(3)}\frac{\partial}{\partial H}\bigl(\widehat P_{\mathrm{diss}}(RF,RH)\bigr)\,d\mu\\
=\int_{\SO(3)}R^T\widehat T_{\mathrm{Rd}}(RF,RH)\bigr)\,d\mu=\widehat T_{\mathrm{Rd}}(F,H)
\end{multline*}
by \eqref{AIM visqueux 1}.

We now turn to frame-indifference. First of all, $\overline P_{\mathrm{diss}}$ satisfies \eqref{AIM dissipation 1} by construction. Secondly, we let $M=\Sym(HF^{-1})$ and $W=\Skew(HF^{-1})$, and for $s\in[0,1]$, we let $H_s=MF+sWF$. Then \eqref{AIM visqueux 2} implies that $\widehat T_{\mathrm{Rd}}(F,H_s)=\widehat T_{\mathrm{Rd}}(F,H)$. Now,
$$
\frac{d}{ds}\bigl(\widehat P_{\mathrm{diss}}(F,H_s)\bigr)=\widehat T_{\mathrm{Rd}}(F,H_s):(WF)=
\widehat T_{\mathrm{Rd}}(F,H):(WF),
$$
so that $\widehat P_{\mathrm{diss}}(F,H_s)=\widehat P_{\mathrm{diss}}(F,MF)+s\widehat T_{\mathrm{Rd}}(F,H):(WF)$. In particular, for $s=1$, we obtain
\begin{align*}
\widehat P_{\mathrm{diss}}(F,H)&=\widehat P_{\mathrm{diss}}(F,MF)+\widehat T_{\mathrm{Rd}}(F,H):(WF)\\
&=\widehat P_{\mathrm{diss}}(F,MF)+\bigl(\widehat T_{\mathrm{Rd}}(F,H)F^T\bigr):W.
\end{align*}
By construction, the mechanical part of the Clausius-Planck inequalities is satisfied, hence by Proposition \ref{incroyable},
 $\widehat T_{\mathrm{Rd}}(F,H)F^T$ is symmetric. Since $W$ is skew-symmetric, the last inner product vanishes. 
Therefore, for all $R\in\SO(3)$, 
$
\widehat P_{\mathrm{diss}}(RF,RH)=\widehat P_{\mathrm{diss}}(RF,RMF)
$
and integrating over $\SO(3)$, we obtain
$$
\overline P_{\mathrm{diss}}(F,H)=\int_{\SO(3)}\widehat P_{\mathrm{diss}}(RF,RMF)\,d\mu
=\overline P_{\mathrm{diss}}(F,MF).$$
Therefore, equation \eqref{AIM dissipation 2} is also satisfied and $\overline P_{\mathrm{diss}}$ is frame-indifferent.
\end{proof}
\begin{remark}As noted earlier and contrarily to what happens in hyperelasticity, a dissipation potential that yields a frame-indifferent  dissipative stress cannot be assumed to be itself frame-indifferent. This is even the case when the mechanical part of the Clausius-Planck inequalities is satisfied. Indeed, the non frame-indifferent potential $\widehat P_{\mathrm{diss}}(F,H)=F_{11}$ yields the perfectly frame-indifferent and Clausius-Planck compliant $\widehat T_{\mathrm{Rd}}(F,H)=0$.
\end{remark}

\subsection{Material symmetries}\label{section symetries}
We refer to the lucid discussion of material symmetries for simple materials in \cite{Truesdell-Noll}, which applies almost unchanged here when there are no internal variables, even though these authors allow symmetries with negative determinant, which we prefer not to consider as they are not physically feasible, see also \cite{Gurtin} in this respect. Material symmetries for models with internal variables must also be considered on a case-by-case basis. 

Material symmetry is a Lagrangian concept. Let us be given a reference configuration $\Omega$, $X_0\in\Omega$ and $S\in \M_3^+$. We consider another reference configuration $\Omega'$ such that $X_0\in \Omega'$ and a diffeomorphism $\Psi_S\colon\Omega'\to\Omega$ such that $\Psi_S(X)=X_0+S(X-X_0)+o(\|X-X_0\|)$. Then $S$ is said to be a material symmetry at $X_0$ with respect to the reference configuration $\Omega$ if, for any deformation $\phi$, the Cauchy stress tensor corresponding to any deformation equal to $\phi\circ\Psi_S$ (composition between understood in the spatial variable) in a neighborhood of $X_0$ is equal to that of $\phi$. This can be expressed as 
$$
T_{\mathrm{R}}^{\phi\circ\Psi_S}(X_0,t)=T_{\mathrm{R}}^\phi(X_0,t)\cof S,
$$
with self-explanatory notation.

It is well known that the set of such material symmetries is a subgroup of $\M_3^+$, that this subgroup should actually be a subgroup of $\SL(3)$, so that $\cof S=S^{-T}$, and that material symmetry groups corresponding to different reference configurations are conjugate to one another. These considerations yield the classical classification of materials as 
\begin{itemize}
\item
solid (at point $X_0$) if the material symmetry group is included in a conjugate of $\SO(3)$, or equivalently that there is a reference configuration in which this group is included in $\SO(3)$, 
\item
isotropic (at point $X_0$) if the material symmetry group contains a conjugate of $\SO(3)$, or equivalently that there is a reference configuration in which this group contains $\SO(3)$, 
\item
fluid (at point $X_0$) if the material symmetry group is equal to $\SL(3)$.
\end{itemize}
In the latter case, $\SL(3)$ is the kernel of the determinant mapping, hence a normal subgroup of $\M_3^+$, which is thus equal to all its conjugates.

It is a surprisingly little known result that $\SO(n)$ is a maximal subgroup of $\SL(n)$, or that $O(n)$ is a maximal subgroup of $\SL^\pm(n)$, without any topological assumption, such as maximal compact or maximal closed, see \cite{Brauer} for the case of $\mathrm{O}(n)$ and $\SO(n)$ and \cite{Noll} for $\mathrm{O}(n)$. Therefore, an isotropic material that is not solid is necessarily fluid. Such materials as liquid crystals are neither fluid, nor solid, nor isotropic, see~\cite{Wangsubfluids}.

In terms of constitutive laws, the fact that $S$ is a material symmetry is expressed as
$$
\begin{aligned}
\widehat T_{\mathrm{R}}(FS,HS,\Theta)&=\widehat T_{\mathrm{R}}(F,H,\Theta)\cof S,\\
\widehat A_m(FS,\Theta)&=\widehat A_m(F,\Theta),\\
\widehat P_{\mathrm{diss}}(FS,HS,\Theta)&=\widehat P_{\mathrm{diss}}(F,H,\Theta),
\end{aligned}
$$
for all $F,H,\Theta$. 

According to \cite{Smith}, see also \cite{Wang}, the Cauchy stress tensor of an isotropic thermo-visco-elastic material without internal variables can be expressed in terms of isotropic invariants and generating  functions, in the spirit of the Rivlin-Ericksen theorem in nonlinear elasticity. Letting $B=FF^T$ and $M=\Sym(HF^{-1})$, the isotropic invariants of $(B,M)$ are
\begin{multline*}
i(B,M)=\bigl(\tr B, \tr (B^2), \tr (B^3), \tr M, \tr (M^2),\tr(M^3),\\
\tr(BM),\tr(B^2M),\tr(BM^2),\tr(B^2M^2)\bigr),
\end{multline*}
and the generating functions are
$$
\mathcal{Z}(B,M)=\bigl(I, B, B^2,M,M^2,BM+MB, B^2M+MB^2,BM^2+M^2B\bigr).
$$
Then
$$
\check \sigma(F,H,\Theta)=\sum_{k=1}^8\beta_k\bigl(i(B,M)),\Theta\bigr)\mathcal{Z}_k(B,M),
$$
where $\check \sigma$ is defined in Proposition \ref{AIM eulerien visqueux} and the scalar-valued functions  $\beta_k$ are arbitrary.

\subsection{Frame-indifference for the heat flux and thermal symmetries}
The heat flux also obeys frame-indifference, in the form \eqref{AIM Q lag} for instance. 
To be compatible with this principle, when there are no internal variables, the constitutive law in Eulerian and Lagrangian forms, with Lagrangian variables, must satisfy
$$\widehat q_{\mathrm L}(RF,\Theta,G)=R\widehat q_{\mathrm L}(F,\Theta,G),\quad\widehat Q(RF,\Theta,G)=\widehat Q(F,\Theta,G),$$
for all $F\in \M_3^+$, $\Theta\in\R_+^*$ and $G\in\R^3$. This is equivalent to the fact that $\widehat Q$ depends on $F$ only via
$C=F^TF$,
$\widehat Q(F,\Theta,G)=\widetilde Q(C,\Theta,G)$.
This is the case of the classical Fourier law, which is thus frame-indifferent (as is also even more apparent on its Eulerian form).

We can also consider thermal symmetries $S\in \SL(3)$ in a manner similar to material symmetries, which are reflected in the constitutive law by 
$$\widehat q_{\mathrm L}(FS,\Theta,S^TG)=\widehat q_{\mathrm L}(F,\Theta,G),\quad\widehat Q(FS,\Theta,S^TG)=(\cof S)^T\widehat Q(F,\Theta,G).$$
There is no a priori reason for thermal symmetries to always be the same as material symmetries. 

It is fairly clear that the classical Fourier law is isotropic. It is actually thermally fluid, \emph{i.e.} with thermal symmetry group $\SL(3)$, thus also appropriate for material fluids.

The following result characterizes all frame-indifferent and isotropic heat flux constitutive laws. Earlier works, such as \cite{Truesdell-Noll}, \cite{Smith} or \cite{Wang}, assume $\mathrm{O}(3)$ symmetry instead of $\SO(3)$ symmetry, hence do not include all possible laws.

Let us introduce some notation. As usual,  $\iota(B)$ denotes the triple of principal invariants of $B=FF^T$. For $K\in \R^3$, we also set $\iota(B,K)=\bigl(\|K\|^2,\|BK\|^2,K\cdot BK,s(B,K)\bigr)\in \R_+^3\times\{-1,0,1\}$, with $s(B,K)=\mathrm{sign}(\prod_{i=1}^3K\cdot v_i)$ if $B$ has three distinct eigenvalues $0<\lambda_1<\lambda_2<\lambda_3$ and $(v_1,v_2,v_3)$ is a right-handed orthonormal basis of corresponding eigenvectors of $B$, with the convention that $\mathrm{sign}(0)=0$, and $s(B,K)=0$ otherwise. It  can be checked that $s$ is actually a function of the pair $(B,K)$ in spite of the multiplicity of possible basis choices. We use the $\wedge$ notation for vector products. 

Representation formulas for isotropic and frame-indifferent constitutive laws are better written for the Eulerian quantities expressed with Lagrangian variables, since both descriptions are involved at the same time.

\begin{proposition}\label{chaleur isotrope}The constitutive law of a frame-indifferent and isotropic heat flux takes the form
$$\widehat q_{\mathrm L}(F,\Theta,G)=\widehat q_{\mathrm{L,iso}}(B,\Theta, F^{-T}G),$$
where 
\begin{multline}\label{forme chaleur isotrope}
\widehat q_{\mathrm{L,iso}}(B,\Theta, K)=\alpha_0\bigl(\iota(B),\Theta,\iota(B,K)\bigr)K+\alpha_1\bigl(\iota(B),\Theta,\iota(B,K)\bigr)BK\\+\alpha_2\bigl(\iota(B),\Theta,\iota(B,K)\bigr)K\wedge BK
\end{multline}
and the functions  $\alpha_i\colon\R_+^7\times\{-1,0,1\}\to\R$ are arbitrary.
\end{proposition}

\begin{proof}We start with isotropy. Given $F_1,F_2\in\M_3^+$ and $G_1,G_2\in\R^3$, there exists $R\in \SO(3)$ such that $F_2R=F_1$ and $R^TG_2=G_1$ if and only if first $F_1F_1^T=F_2F_2^T$, and then $F_2^{-T}G_2=F_1^{-T}G_1$, hence the existence of the function $\widehat q_{\mathrm{L,iso}}$. Conversely, any such function gives rise to an isotropic constitutive law. 

Setting $F^{-T}G=K$, we then express frame-indifference with this representation, which reads
\begin{equation}\label{iso puis AIM}
\widehat q_{\mathrm{L,iso}}(RBR^T,RK)=R\widehat q_{\mathrm{L,iso}}(B,K)
\end{equation}
for all $R\in\SO(3)$, $B\in\Sym_3^+$ and $K\in\R^3$.

%La fonction $\widehat q_{L,\mathrm{iso}}$ donnée par l'isotropie est a priori une fonction de 9 variables scalaires indépendantes à valeurs dans $\R^3$. Grâce à l'AIM, on souhaite diminuer au maximum ce nombre, ici à 6 variables scalaires  indépendantes positives, et une valeur appartenant à $\{-1,0,1\}$, également à valeurs dans $\R^3$. 

We first establish an intermediate representation  
\begin{equation}\label{intermediaire}
\widehat q_{\mathrm{L,iso}}(B, K)=\gamma_0(B,K)K+\gamma_1(B,K)BK+\gamma_2(B,K)K\wedge BK,
\end{equation}
where the functions $\gamma_i\colon \Sym_3^+\times\R^3\to \R$ are such that $\gamma_i(RBR^T,RK)=\gamma_i(B,K)$. 

Assume that \eqref{iso puis AIM} holds. There are three different cases.

\begin{enumerate}[leftmargin=0pt,itemindent=25pt]
\item $K=0$. There are three rotations of independent axes such that $RBR^T=B$, which implies $\widehat q_{\mathrm{L,iso}}(B,0)=0$. In this case, we set $\gamma_i(B,0)=0$ which have the required invariance.
\item $K\wedge BK\neq 0$. Then $(K,BK,K\wedge BK)$ is basis of $\R^3$ and we let $\gamma_i(B,K)$ be the coordinates of $\widehat q_{\mathrm{L,iso}}(B, K)$ in this basis. Since $(RK,(RBR^T)RK,RK\wedge (RBR^T)RK)=R(K,BK,K\wedge BK)$, they also have the required invariance.
\item  $K\neq 0$ and $K\wedge BK= 0$. In this case, $K$ is an eigenvector of $B$. Let $R_0$ be the rotation of axis directed by $K$ and angle $\pi$, which leaves $B$ and $K$ invariant. It follows that $\widehat q_{\mathrm{L,iso}}(B,K)=R_0\widehat q_{\mathrm{L,iso}}(B,K)$, which means that $\widehat q_{\mathrm{L,iso}}(B,K)$ belongs to the axis of $R_0$, \emph{i.e.}, is colinear with $K$. We thus set $\gamma_0(B,K)=\frac1{\|K\|^2}\bigl(\widehat q_{\mathrm{L,iso}}(B,K)\cdot K\bigr)$ and $\gamma_1(B,K)=\gamma_2(B,K)=0$. These also have the required invariance since for any rotation $R$, $RK$ is an eigenvector of $RBR^T$. 
\end{enumerate}

Conversely, a function of the form \eqref{intermediaire} is frame-indifferent, \emph{i.e.}, satisfies \eqref{iso puis AIM}.

The question now is thus to characterize all functions $\gamma\colon \Sym_3^+\times\R^3\to \R$ such that $\gamma(RBR^T,RK)=\gamma(B,K)$ for all $R$, $B$, and $K$, using the smallest number of independent scalar variables. We are going to show that given $B_1$, $B_2$, $K_1$, and $K_2$, there exists a rotation $R$ such that $B_2=RB_1R^T$ and $K_2=RK_1$ if and only if $\iota(B_1)=\iota(B_2)$ and $\iota(B_1,K_1)=\iota(B_2,K_2)$. 

First the necessary condition. If $B_2=RB_1R^T$, then $\iota(B_1)=\iota(B_2)$. If furthermore $K_2=RK_1$, then $\|K_2\|=\|K_1\|$, $\|B_2K_2\|=\|RB_1R^TRK_1\|=\|B_1K_1\|$ et $K_2\cdot B_2K_2=RK_1\cdot RB_1R^TRK_1=K_1\cdot B_1K_1$. Lastly, if $B_1$ has three distinct eigenvalues, $0<\lambda_1<\lambda_2<\lambda_3$, so does $B_2$. We let $(v_1,v_2,v_3)$, resp.\ $(w_1,w_2,w_3)$, be right-handed orthonormal bases of eigenvectors for $B_1$, resp.\ $B_2$, in the same order as the eigenvalues. Writing $K_1=\sum_{i=1}^3(K_1\cdot v_i )v_i$, it follows that $K_2=\sum_{i=1}^3(K_1\cdot v_i )Rv_i=\sum_{i=1}^3(K_2\cdot Rv_i )Rv_i$. Now $Rv_i=\varepsilon_i w_i$ with $\varepsilon_i=\pm1$, therefore $K_1\cdot v_i =\varepsilon_i K_2\cdot w_i $. Since the two bases are right-handed, there are only two possibilities: either $\varepsilon_i=1$ for all $i$, or two of them are $-1$ and the remaining one is $1$. In both cases, $\prod_{i=1}^3K_1\cdot v_i =\prod_{i=1}^3K_2\cdot w_i $, and thus $s(B_1,K_1)=s(B_2,K_2)$. This completes the proof that $\iota(B_1,K_1)=\iota(B_2,K_2)$.

We now turn to the sufficient condition. Assume that $\iota(B_1)=\iota(B_2)$ and $\iota(B_1,K_1)=\iota(B_2,K_2)$. We must construct an appropriate rotation $R$. Since $\iota(B_1)=\iota(B_2)$, there are several rotations $R$ such that $B_2=RB_1R^T$. If $K_1=0$, then $K_2=0$ and any such rotation works. Assuming $K_1$ and $K_2$ nonzero, we then discuss according to the common multiplicity of the eigenvalues of $B_1$ and $B_2$.

\begin{enumerate}[leftmargin=0pt,itemindent=25pt]
\item Three distinct eigenvalues $\lambda_i$. Let $(v_i)$ and $(w_i)$ be right-handed orthonormal eigenvector bases as before. There exists a unique rotation $R_0$ such that $R_0v_i=w_i$, and thus $B_2=R_0B_1R_0^T$. The hypothesis $\iota(B_1,K_1)=\iota(B_2,K_2)$ first implies that
$$
\sum_{i=0}^3\lambda_i^k(K_2\cdot w_i)^2=\sum_{i=0}^3\lambda_i^k(K_1\cdot v_i)^2, \quad k=0,1,2.
$$
This is an invertible Vandermonde system, with unique solution $(K_2\cdot w_i)^2=(K_1\cdot v_i)^2$, $i=1,2,3$. Consequently, $|K_2\cdot w_i|=|K_1\cdot v_i|$, $i=1,2,3$. Therefore $K_2\cdot w_i=\varepsilon_i K_1\cdot v_i$ with $\varepsilon_i=\pm1$ being uniquely determined when paired with nonzero terms. 

First subcase, the three terms are nonzero. The condition $s(B_1,K_1)=s(B_2,K_2)$ implies that either $\varepsilon_i=1$ for all $i$, or two of them are $-1$ and the remaining one is $1$. In both cases, $\varepsilon_iK_1\cdot v_i=(R^0)^TK_1\cdot v_i$ where $R^0$ is the rotation defined by $R^0v_i=\varepsilon_iv_i$. Consequently $K_2=R_0(R^0)^TK_1$ and we still have $B_2=R_0R^0B_1(R_0R^0)^T$ since $B_1=R^0B_1(R^0)^T$. 

Second subcase, two nonzero terms, one zero term, \emph{i.e.}, without loss of generality $K_1\cdot v_1\neq0$, $K_1\cdot v_2\neq0$ and $K_1\cdot v_3=0$. Then $\varepsilon_1$ and $\varepsilon_2$ are determined and we set $\varepsilon_3=\varepsilon_1\varepsilon_2$ to define an adequate rotation $R^0$.

Last subcase, one nonzero term, two zero terms, \emph{i.e.}, $K_1\cdot v_1\neq0$, $K_1\cdot v_2=K_1\cdot v_3=0$. Then $K_1$ is colinear with $v_1$ and $K_2$ is colinear with $w_1$, so that $K_2=\varepsilon_1R_0K_1$. We take $R^0$ defined by $(\varepsilon_1,\varepsilon_1,1)$. 

\item Two equal eigenvalues, the third one being distinct, without loss of generality, $\lambda_1=\lambda_2\neq\lambda_3$ with corresponding right-handed orthonormal eigenvectors $v_i$ and $w_i$, $i=1,2,3$, for $B_1$ and $B_2$. All rotations that map $v_3$ on $\pm w_3$ make $B_1$ and $B_2$ conjugate. In this case, the Vandermonde system is not invertible. It is however equivalent to $(K_2\cdot w_1)^2+(K_2\cdot w_2)^2=(K_1\cdot v_1)^2+(K_1\cdot v_2)^2$ and $(K_2\cdot w_3)^2=(K_1\cdot v_3)^2$. In other words, the projection of $K_1$ on $v_3^\bot$ has the same norm as the projection of $K_2$ on $w_3^\bot$, and $K_2\cdot w_3=\varepsilon_3K_1\cdot v_3$. There is thus a rotation that maps $K_1$ on $K_2$ while mapping $v_3$ on $\pm w_3$.

\item Three equal eigenvalues. Then $B_2=B_1=\lambda I$ and the condition $\|K_1\|=\|K_2\|$ is sufficient for the existence of an appropriate rotation.
\end{enumerate}
In all three cases, there is a rotation $R$ such that $B_2=RB_1R^T$ and $K_2=RK_1$.\end{proof}

Let us emphasize once again that the form given in Proposition \ref{chaleur isotrope} is more general than the one that can be found for instance in \cite{Smith}, which does not include $\widehat q_{\mathrm{L,iso}}(B,\Theta, K)=K\wedge BK$ for example, an admittedly physically strange heat flux, that is nonetheless frame-indifferent, isotropic and satisfies the Clausius-Planck inequality. Indeed, $\widehat q_{\mathrm{L,iso}}(B,\Theta, K)\cdot K=0$.

More generally, the thermal part of the Clausius-Planck inequalities is satisfied if and only if
$$
\alpha_0\bigl(\iota(B),\Theta,\iota(B,K)\bigr)\|K\|^2+\alpha_1\bigl(\iota(B),\Theta,\iota(B,K)\bigr)BK\cdot K\le 0,$$
which is in particular the case if $\alpha_0$ and $\alpha_1$ are nonpositive.

Let us also determine all thermally fluid heat fluxes. This is a particular case of the previous result, but it is easier not to start from \eqref{forme chaleur isotrope}.
\begin{proposition}The constitutive law of a thermally fluid heat flux takes the form
\begin{equation}\label{representation chaleur fluide eulerienne}
\widehat q_{\mathrm L}(F,\Theta,G)=-\check k_{\mathrm{L},\fluide}(\det F,\Theta,\|F^{-T}G\|)F^{-T}G,
\end{equation}
where $\check k_{\mathrm{L},\fluide}\colon (\R_+^*)^2\times\R_+\to\R$ is arbitrary.

In Eulerian variables, this also reads
\begin{equation}\label{representation chaleur fluide eulerienne eulerienne}
\widehat q(f,\theta,g)=-\check k_{\mathrm{E},\fluide}(\rho,\theta,\|g\|)g,
\end{equation}
with $\check k_{\mathrm{E},\fluide}\colon (\R_+^*)^2\times\R_+\to\R$.

Such a law satisfies the Clausius-Planck inequality if and only if the scalar functions $\check k_{\mathrm{L},\fluide}$ and $\check k_{\mathrm{E},\fluide}$ are nonnegative.
\end{proposition}

\begin{proof}We start with fluidity. Given $F_1,F_2\in\M_3^+$ and $G_1,G_2\in\R^3$, there exists $S\in \SL(3)$ such that $F_2S=F_1$ and $S^TG_2=G_1$ if and only if first $\det F_1=\det F_2$, and then $F_2^{-T}G_2=F_1^{-T}G_1$. Therefore, we can write
 $$\widehat q_{\mathrm L}(F,G)=\check q_{\mathrm L}(\det F,F^{-T}G)$$
  with $\check q_{\mathrm L}\colon \R_+^*\times \R^3\to\R^3$.
 Conversely, any such function gives rise to a fluid constitutive law for the heat flux. 
 
 Setting $F^{-T}G=K$, we then express frame-indifference with this representation, which reads
$$\widehat q_{\mathrm L}(RF,G)=\check q_{\mathrm L}(\det F,RK)=R\check q_{\mathrm L}(\det F,K),$$
which says that $K\mapsto \check q_{\mathrm L}(J,K)$ is an objective function on $\R^3$ for all $J$, see \cite{Gurtin}. It is well known that this is equivalent to having
$$ \check q_{\mathrm L}(J,K)=-\check k_{\mathrm{L},\fluide}(J,\|K\|)K$$
where $\check k_{\mathrm{L},\fluide}$ is scalar-valued, which is exactly \eqref{representation chaleur fluide eulerienne}.
We then use $f=F^{-1}$, $g=F^{-T}G$, $\rho=\Rho/J$ to rewrite it as \eqref{representation chaleur fluide eulerienne eulerienne}. 

Finally
$$\check q_{\mathrm L}(J,K)\cdot K=-\check k_{\mathrm{L},\fluide}(J,\|K\|)\|K\|^2
$$
so that the thermal part of Clausius-Planck is satisfied if and only if $\check k_{\mathrm{L},\fluide}(J,\|K\|)\ge 0$ for all $J$ and $K$.\end{proof}

It thus turns out that all fluid, frame-indifferent heat fluxes are actually nonlinear Fourier laws.

We alluded earlier to the use of a diffusion potential to construct heat fluxes that satisfy the thermal part of the Clausius-Planck inequalities. Let us state this precisely, together with frame-indifference and thermal symmetry conditions. 
\begin{proposition}\label{potentiel de diffusion}
Let $\widehat P_{\mathrm{diff}}\colon\M_3^+\times\R_+^*\times \R^3\to\R_+$ be a function which is concave with respect to its third argument and such that $\widehat P_{\mathrm{diff}}(F,\Theta,0)=0$ for all $F$ and $\Theta$. Then $\widehat Q=\frac{\partial \widehat P_{\mathrm{diff}}}{\partial G}$ defines a heat flux that satisfies Clausius-Planck.

If this potential is frame-indifferent, \emph{i.e.}, $\widehat P_{\mathrm{diff}}(RF,\Theta,G)=\widehat P_{\mathrm{diff}}(F,\Theta,G)$, so is its associated heat flux. If this potential has a thermal symmetry $S\in\SL(3)$, \emph{i.e.}, $\widehat P_{\mathrm{diff}}(FS,\Theta,S^TG)=\widehat P_{\mathrm{diff}}(F,\Theta,G)$, its associated heat flux also has the symmetry $S$.
\end{proposition}

\begin{proof}Clear.\end{proof} 

For example, the diffusion potential $\widehat P_{\mathrm{diff}}(F,G)=-\frac k2G^TC^{-1}G$ with $k>0$ gives rise to the classical Fourier law $\widehat Q(F,G)=-kC^{-1}G$.

\section{Examples of thermo-visco-elastic materials}\label{des exemples}

We now give a few examples of materials, old and new, that fall within our global framework.

\subsection{Thermo-elastic materials}These are of course the simplest of all with no internal variables and no dissipative stress. They are solely characterized by their frame-indifferent free energy $\widehat A_m$, with $\widehat T_{\mathrm{R}}(F,\Theta)=\Rho\frac{\partial\widehat A_m}{\partial F}(F,\Theta)$ and $\widehat S_m(F,\Theta)=-\frac{\partial \widehat A_m}{\partial \Theta}(F,\Theta)$, and frame-indifferent heat flux $\widehat Q$. The internal dissipation is zero. If the material is in thermal equilibrium at all times, \emph{i.e.}, $\nabla_X\Theta=0$, then the Clausius-Duhem inequality is an equality and all evolutions are reversible: heat and mechanical energy can be transformed into one another in both directions without any loss.

When the free energy is split in the form $\widehat A_m(F,\Theta)=\widehat W_m(F)+\widehat V_m(\Theta)$, the model decouples into a nonlinear elasticity model on the side of stresses and dynamics without any thermal effect, and a nonlinear heat equation for the temperature with no mechanical source term, as there is no internal dissipation, even though the heat flux may still depend on $F$.

Thermo-elastic materials can have any possible material symmetry, for instance be solid or fluid. It is easy to see that the free energy of a thermo-elastic fluid is of the form $\widehat A_m(F,\Theta)=\widehat\Psi_m(\det F,\Theta)$ and that the Cauchy stress is a pure pressure $\sigma=-p(\rho,\theta) I$ in the Eulerian description. This includes perfect gases $\widehat\Psi_m(J,\Theta)=-r_{pg}\Theta\ln J+\widehat V_m(\Theta)$.

\subsection{Kinematically viscous materials}At the other end of the spectrum are materials with no elastic stress at all, $\widehat A_m(F,\Theta)=\widehat V_m(\Theta)$,  an entirely dissipative stress $\widehat T_{\mathrm{Rd}}(F,H,\Theta)$ and still no internal variables.

Such materials can have different symmetries. For instance, the following is a somewhat artificial solid example:  let $\nu\colon\R_+^*\times\R_+^*\to\R_+^*$ be strictly increasing with respect to its first variable and take  $\widehat T_{\mathrm{Rd}}(F,H,\Theta)=\nu(\tr C,\Theta)\det F\bigl(HC^{-1}+F^{-T}H^TF^{-T}\bigr).$ This material satisfies the mechanical part of the Clausius-Planck inequalities, is frame-indifferent and a matrix $S$ is a material symmetry if and only if $\nu(\tr(S^T CS),\Theta)=\nu(\tr C,\Theta)$ for all $C\in\Sym_3^+$, or $\tr(S^T CS)=\tr(C)$ for all such $C$. In particular, $SS^T-I\in C^\bot$ for all $C$ and thus $SS^T=I$, which shows that $S\in SO(3)$. This material is of course isotropic. 

When $\nu$ is instead a strictly positive function of $\det F$ and  $\Theta$, then the corresponding material is a compressible frame-indifferent viscous fluid, still satisfying the mechanical part of the Clausius-Planck inequalities. When $\nu$ is a constant, the material is a Newtonian compressible fluid, the dynamics equations of which in the Eulerian description are the compressible Navier-Stokes equations.

More generally, all viscous fluids in this family have a constitutive law for the Cauchy stress in the Eulerian description of the form
$$
\check \sigma(\rho, d,\theta)=\beta_0(\rho,\iota(d),\theta)I+\beta_1(\rho,\iota(d),\theta)d+\beta_2(\rho,\iota(d),\theta)d^2,
$$
where $\beta_i$ are arbitrary real-valued functions and $\iota(d)$ is the triple of principal invariants of $d$, by a direct application of the Rivlin-Ericksen theorem. Such non-Newtonian fluids are known as compressible Reiner-Rivlin fluids, \cite{Reiner}-\cite{Rivlin}. The Clausius-Planck inequality then demands that
$$
\Bigr(\sum_{i=0}^2\beta_i(\rho,\iota(d),\theta)d^i\Bigl):d\ge 0.
$$
It is as a rule strict and leads to irreversibility, even in thermal equilibrium.

By adding a term $\widehat\Psi_m(\det F,\Theta)$ to the free energy, we obtain thermo-visco-elastic fluids, for which the above inequality must be slightly adapted. 

\subsection{A family of nonlinear 3d Maxwell models}\label{modeles maxwell 3d}
We now present a family of materials that do not seem to be found in the literature to the best of our knowledge. It is intended to provide three-dimensional, frame-indifferent, nonlinear generalizations of the Maxwell rheological model, a zero-dimensional model which consists in a linearly elastic spring and a linearly viscous dashpot placed in series, a model that exhibits stress relaxation.  There are other attempts at extending the Maxwell and generalized Maxwell rheological models (the latter with stress relaxation and creep) to a full 3d setting, see for example \cite{Holtzapfel-Simo}. 

In the Maxwell model, the total stretching of the system is denoted $\varepsilon$, that of the dashpot $\gamma$, so that the stretching of the spring is $\varepsilon-\gamma$. If $\mu>0$ denotes the stiffness of the spring and $\nu>0$ the viscosity of the dashpot, then there is an elastic force $\mu(\varepsilon-\gamma)$ and a viscous friction force $-\nu\dot\gamma$ (we use the dot for the usual time derivative, there is no Eulerian/Lagrangian distinction here). Since there is no mass between the spring and the dashpot, these forces are equilibrated at all times and it is fairly clear that if $\varepsilon$ has a prescribed constant value in time $\bar\varepsilon$, then stress relaxation will occur for any initial values of the stretchings, \emph{i.e.}, the resultant force applied to each end of the system, $\mu(\bar\varepsilon-\gamma)=\nu\dot\gamma$ or its opposite depending on which end it is applied to, will decay to $0$ exponentially in time as the spring settles back to its natural length, \emph{i.e.} to zero stretch, while being restrained by the dashpot.

In order to fit the Maxwell model into our thermomechanical framework, it is very natural to consider  $\varepsilon$ as a thermodynamic variable playing the role of $F$, and $\gamma$ as an internal variable playing the role of $\Xi$, with no temperature (or decoupled temperature). Indeed, the system should be considered to be installed inside a black box, of which only $\varepsilon$ is observable. None of the two stretches happening inside are observable.
Taking as free energy the elastic energy of the spring, 
  $\widehat A_m(\varepsilon,\gamma)=\frac\mu2(\varepsilon-\gamma)^2$, 
 and as right-hand side for the ordinary differential equation $\dot\gamma=\widehat K(\varepsilon,\gamma)$, 
  $\widehat K(\varepsilon,\gamma)=\frac{\mu}{\nu}(\varepsilon-\gamma)$, and applying the results of a very degenerate kind of Coleman-Noll procedure,
  we recover exactly the Maxwell model. There is also a dissipation potential $\widehat P_{\mathrm{diss}}(\varepsilon,\lambda)=\frac\kappa2\lambda^2$,  with $\kappa=\frac{1}{\nu}$, where $\lambda$ plays the role of $\Lambda$. 
  
  It should be noted that the viscous behavior of the Maxwell model is not kinematical, since the viscous effect is not a function of the observable deformation rate $\dot\varepsilon$. 
  
  It is fairly easy to devise thermodynamically sound nonlinear, zero-dimensional versions of the Maxwell model by considering more general free-energies and flow rules generated by more general dissipation potentials, and also by adding temperature as well.
  
We are more interested here in extending the kind of behavior of the Maxwell model to a 3d setting, which should be fully nonlinear and frame-indifferent. The Maxwell model is based on an additive decomposition of strains, $\varepsilon=(\varepsilon-\gamma)+\gamma$, which will not do for our purposes. We thus turn to a multiplicative decomposition of strains, a very common idea in many contexts such as visco-elastic porous media \cite{Markert} or plasticity \cite{DavoliFrancfort}, \cite{Mielke}, see also \cite{Le Tallec} for a simpler viscoelastic version.

The simplest assumption is thus to take $F\in\M_3^+$ as thermodynamic variable, no dissipative stress, no temperature, and an internal variable $F_i\in\M_3^+$. Now $F$ will take the place of $\varepsilon$, the observable strain, and $F_i$ that of $\gamma$, a sort of internal viscous strain. Of course, we still have $F=\nabla_X\phi$, but $F_i$ is not the gradient of a deformation in general, and should not be interpreted that way.

Given any frame-indifferent nonlinearly elastic stored energy function $\widehat W$, we consider the free energy constitutive law
\begin{equation}\label{energie libre Maxwell 3D}
\widehat A_m(F,F_i)=\frac1\Rho \widehat W(FF_i^{-1})=\frac1\Rho \widehat W(F_e),
\end{equation}
where $F_e=FF_i^{-1}$, which thus acts as a sort of internal elastic strain, without being the gradient of a deformation either. Without loss of generality, we let $\Rho=1$. Since $\widehat W$ is assumed to be frame-indifferent, the free energy inherits a kind of frame-indifference in the form 
$$\widehat A_m(RF,F_i)=\widehat A_m(F,F_i),$$
but we will return to frame-indifference issues later on. 

In effect, we are considering in equation \eqref{energie libre Maxwell 3D} a multiplicative decomposition of the strain of the form $F=F_eF_i$, where $F_i$ is considered as the internal variable. There is much debate in the literature, in particular concerning plasticity, about the order in which such a decomposition should be made. In our context, where frame-indifference and the related symmetry of the Cauchy stress are of primary concern, setting $F=F_iF_e$ would not be appropriate.

We assume that there is no dissipative part of the  stress, $\widehat T_{\mathrm{Rd}}=0$. The Coleman-Noll procedure then implies that
\begin{equation}\label{TR Maxwell}
\widehat T_{\mathrm{R}}(F,F_i)=\frac{\partial\widehat A_m}{\partial F}(F,F_i)=\frac{\partial\widehat  W}{\partial F_e}(FF_i^{-1})F_i^{-T},
\end{equation}
to be used in the dynamics equation, or in a quasistatic version thereof.  The resulting models  thus do  not describe kinematically viscous materials, even though there are internal viscous effects at work.

To complete the model in our general framework, we need an ordinary differential equation for the internal variable of the form
\begin{equation}\label{edo Maxwell}
\frac{\partial F_i}{\partial t}(X,t)=\widehat K(F(X,t),F_i(X,t)).
\end{equation}
Since
\begin{equation}\label{force thermodynamique}
\frac{\partial\widehat A_m}{\partial F_i}(F,F_i)=-F_i^{-T}F^T\frac{\partial\widehat  W}{\partial F_e}(FF_i^{-1})F_i^{-T},
\end{equation}
the mechanical part of the Clausius-Planck inequalities reads
$$
F_i^{-T}F^T\frac{\partial\widehat  W}{\partial F_e}(FF_i^{-1})F_i^{-T}: \widehat K(F,F_i)\ge0.
$$
  
This inequality can be ensured in a systematic way by appealing to Proposition~\ref{potentiel de dissipation cas visqueux et var internes}. Consider a dissipation potential $\widehat P_{\mathrm{diss}}\colon 
 \M_3^+\times \M_3\to\R_+$, convex with respect to its second argument and such that $\widehat P_{\mathrm{diss}}(F,0)=0$. Then  
 \begin{equation}\label{potentiel diss maxwell}
 \widehat K(F,F_i)=-\frac{\partial  \widehat P_{\mathrm{diss}}}{\partial\Lambda}\Bigl(F,\frac{\partial\widehat A_m}{\partial F_i}(F,F_i)\Bigl)
 \end{equation}
  is a  flow rule that satisfies the mechanical part of the Clausius-Planck inequalities. The simplest potential of all is $\widehat P_{\mathrm{diss}}(F,\Lambda)=\frac\kappa2\|\Lambda\|^2$ with $\kappa\ge 0$ and corresponds to the choice 
\begin{equation}\label{Captain Obvious}
\widehat K(F,F_i)=\kappa F_i^{-T}F^T\frac{\partial\widehat  W}{\partial F_e}(FF_i^{-1})F_i^{-T}.
\end{equation}

Before discussing frame-indifference and symmetries, let us note that since $\widehat W$ is assumed to be frame-indifferent, it can be rewritten as
 $\widehat W(F_e)=\widetilde W(C_e)$ with $C_e=F_e^TF_e$ and still $F_e=FF_i^{-1}$. Therefore,
 $$
\widehat T_{\mathrm{R}}(F,F_i)=2FF_i^{-1}\frac{\partial\widetilde  W}{\partial C_e}(F_e^TF_e)F_i^{-T}%=2FF_i^{-1}\frac{\partial\widetilde  W}{\partial C_e}(F_i^{-T}F^TFF_i^{-1})F_i^{-T}.
$$
with $\frac{\partial\widetilde  W}{\partial C_e}$ symmetric. This has an important consequence, namely that the Cauchy stress tensor is automatically symmetric, irrespective of the chosen flow rule. Indeed, 
 $$
F\widehat T_{\mathrm{R}}(F,F_i)^T=2FF_i^{-1}\frac{\partial\widetilde  W}{\partial C_e}(F_e^TF_e)F_i^{-T}F^T=\widehat T_{\mathrm{R}}(F,F_i)F^T.
$$
Choosing a strain decomposition in the reverse order would lead to severe difficulties here, see \cite{DavoliFrancfort}.

Let us now return to the question of frame-indifference. This is a model with an unobservable internal variable, the physical nature of which is furthermore unclear. The general considerations of Section \ref{section AIM visqueux} do not apply, and we need to go back to the initial formulation of the principle of frame-indifference in Section \ref{principe AIM}. We use the same notation with unstarred and starred quantities and obviously, only rotations need to be taken into account since translations are ignored by the model.

Our first observation is that for all $R\in\SO(3)$, $F,F_i\in \M_3^+$,
\begin{equation}\label{AIM on y croit}
\widehat T_{\mathrm{R}}(RF,F_i)=\frac{\partial\widehat  W}{\partial F_e}(RFF_i^{-1})F_i^{-T}=R\frac{\partial\widehat  W}{\partial F_e}(FF_i^{-1})F_i^{-T}=R\widehat T_{\mathrm{R}}(F,F_i),
\end{equation}
by the assumed frame-indifference of $\widehat W$. In other words, the first Piolà-Kirchhoff stress tensor transforms as expected, provided that the internal variable is not affected by superimposed rotations. We thus need an additional hypothesis on the flow rule, namely that
 \begin{equation}\label{AIM et edo}
\widehat K(RF,F_i)=\widehat K(F,F_i),
\end{equation}
and we must also pay attention to an often neglected issue in the context of internal variables, that of the initial conditions for the ordinary differential equation. In accordance with the above observation, they need to be unmodified as well. Finally, we assume $\widehat K$ to be continuous and locally Lipschitz with respect to its second variable, uniformly with respect to its first variable.

\begin{proposition}\label{Maxwell 3D AIM}The 3d Maxwell model \eqref{energie libre Maxwell 3D}--\eqref{edo Maxwell}, with hypothesis \eqref{AIM et edo} and the assumed regularity of $\widehat K$, is  frame-indifferent. 
\end{proposition}

\begin{proof}
Let  $T_{\mathrm{R}}(X,t)$ be the first Piolà-Kirchhoff stress tensor observed at point $X$ and time $t$ when the body undergoes a deformation $\phi$ and $T_{\mathrm{R}}^*(X,t)$ when it undergoes the deformation $R(t)\phi$, where $R$ is an arbitrary $\SO(3)$-valued function. We use the same notation for $F(X,t),F_i(X,t)$ and $F^*(X,t),F_i^*(X,t)$. Of course, $F^*(X,t)=R(t)F(X,t)$. 

In what follows, the material point $X$ is going to be fixed and the only variable is actually the time $t$. For brevity, we thus do not write $X$, it is implicitly where it needs to be.

Now $t\mapsto F(t)$ is given and continuous, hence the right-hand side equation \eqref{edo Maxwell}
 satisfies the hypotheses of the Cauchy-Lipschitz theorem. In particular, we have uniqueness of local solutions to the Cauchy problem
$$
\frac{\partial F_i}{\partial t}(t)=\widehat K(F(t),F_i(t)),\quad F_i(0)=F_{i,0},
$$
for all $F_{i,0}\in\M_3^+$. The same holds for 
$$
\frac{\partial F^*_i}{\partial t}(t)=\widehat K(F^*(t),F^*_i(t)),\quad F^*_i(0)=F^*_{i,0},
$$
for all $F^*_{i,0}\in\M_3^+$. If we assume that $F^*_{i,0}=F_{i,0}$, then $F_i$ and $F^*_i$ are solutions of the same Cauchy problem, by hypothesis \eqref{AIM et edo}. Thus we have $F^*_i=F_i$ and
\begin{multline*}
T_{\mathrm{R}}^{*}(t)=\widehat T_{\mathrm{R}}\bigl(R(t)F(t),F_i^*(t)\bigr)=\widehat T_{\mathrm{R}}\bigl(R(t)F(t),F_i(t)\bigr)\\
=R(t)\widehat T_{\mathrm{R}}\bigl(F(t),F_i(t)\bigr)=R(t)T_{\mathrm{R}}(t),
\end{multline*}
by \eqref{AIM on y croit}, and frame-indifference is satisfied.
\end{proof}

For example, the model that  corresponds to \eqref{Captain Obvious} is frame-indifferent.
When the flow rule is given by a dissipation potential, we have the following characterization.
\begin{proposition}
 If the dissipation potential satisfies 
 \begin{equation}\label{potentiel Maxwell AIM}
 \widehat P_{\mathrm{diss}}(RF,\Lambda)=\widehat P_{\mathrm{diss}}(F,\Lambda),
 \end{equation} 
for all $R\in\SO(3)$, $F\in\M_3^+$, and $\Lambda\in\M_3$, then the 3d Maxwell model is frame-indifferent.
 \end{proposition}
 
 \begin{proof}We go back to equation \eqref{potentiel diss maxwell} and note that
$$
\widehat K(RF,F_i)=-\frac{\partial  \widehat P_{\mathrm{diss}}}{\partial\Lambda}\Bigl(RF,\frac{\partial\widehat A_m}{\partial F_i}(RF,F_i)\Bigl)=-\frac{\partial  \widehat P_{\mathrm{diss}}}{\partial\Lambda}\Bigl(F,\frac{\partial\widehat A_m}{\partial F_i}(RF,F_i)\Bigl),
$$
by hypothesis \eqref{potentiel Maxwell AIM}. Now we have $\widehat A_m(RF,F_i)=\widehat A_m(F,F_i)$, therefore  $\frac{\partial\widehat A_m}{\partial F_i}(RF,F_i)=\frac{\partial\widehat A_m}{\partial F_i}(F,F_i)$ and the conclusion follows.  
 \end{proof}

Material symmetry is studied in a similar fashion. In particular, the initial conditions for the internal variable must be changed according to the symmetry considered.

\begin{proposition}\label{Maxwell 3D symetrie quelconque}Let $\mathcal{S}$ be a subgroup of $\SL(3)$. If we assume that the flow rule satisfies
 \begin{equation}\label{symetrie et edo}
\widehat K(FS,F_iS)=\widehat K(F,F_i)S,
\end{equation}
for all $S\in\mathcal{S}$, and all $F,F_i\in \M_3^+$, then the 3d Maxwell model has material symmetry group $\mathcal{S}$.

If there is a dissipation potential, then \eqref{symetrie et edo} is implied by 
 \begin{equation}\label{potentiel Maxwell fluide}
 \widehat P_{\mathrm{diss}}(FS,\Lambda S^{-T})=\widehat P_{\mathrm{diss}}(F,\Lambda),
 \end{equation}
 for all $S\in\mathcal{S}$, and all $F\in\M_3^+$, $\Lambda\in \M_3$.
\end{proposition}

\begin{proof}With the same notation as in section \ref{section symetries}, only abbreviating $\phi\circ\Psi_S$ as $\phi\circ S$, material symmetry reads
$$T_{\mathrm{R}}^{\phi\circ S}(X,t)=T_{\mathrm{R}}^\phi(X,t)\cof S=T_{\mathrm{R}}^\phi(X,t) S^{-T},$$
for all $\phi$ and $S\in \mathcal{S}$. We again drop $X$ from now on and end up with two Cauchy problems
$$
\left\{
\begin{aligned}
&\frac{\partial F_i^\phi}{\partial t}(t)=\widehat K\bigl(F(t),F_i^\phi(t)\bigr)\\
&F_i^\phi(0)=F_{i,0}^\phi
\end{aligned}
\right.
\text{ and }
\left\{
\begin{aligned}
&\frac{\partial F_i^{\phi\circ S}}{\partial t}(t)=\widehat K\bigl(F(t)S,F_i^{\phi\circ S}(t)\bigr)\\
&F_i^{\phi\circ S}(0)=F_{i,0}^{\phi\circ S}.
\end{aligned}
\right.
$$
Assuming that the initial conditions agree with the symmetry $S$, $F_{i,0}^{\phi\circ S}=F_{i,0}^{\phi}S$, due to \eqref{symetrie et edo} and Cauchy-Lipschitz uniqueness, we deduce that $F_i^{\phi\circ S}=F_i^{\phi}S$. 
Consequently,
\begin{multline*}
T_{\mathrm{R}}^{\phi\circ S}(t)=\widehat T_{\mathrm{R}}\bigl(F(t)S,F_i^{\phi\circ S}(t)\bigr)=\widehat T_{\mathrm{R}}\bigl(F(t)S,F_i^\phi(t) S\bigr)\\
=\frac{\partial\widehat  W}{\partial F_e}(F(t)(F_i^\phi(t))^{-1})(F_i^\phi (t)S)^{-T}=\widehat T_{\mathrm{R}}\bigl(F(t),F_i^\phi(t)\bigr)S^{-T}=T_{\mathrm{R}}^\phi(t)S^{-T},
\end{multline*}
that is to say that $S$ is a material symmetry.

Concerning dissipation potentials, it follows from \eqref{potentiel Maxwell fluide} that
$$
\frac{\partial  \widehat P_{\mathrm{diss}}}{\partial\Lambda}(FS,\Lambda S^{-T})=\frac{\partial  \widehat P_{\mathrm{diss}}}{\partial\Lambda}(F,\Lambda)S.
$$
But equation \eqref{force thermodynamique} implies that
$$
\frac{\partial  \widehat A_m}{\partial F_i}(FS,F_iS)=\frac{\partial  \widehat A_m}{\partial F_i}(F,F_i)S^{-T},
$$
so that \eqref{symetrie et edo} follows.
\end{proof}

\begin{remark}
It is very remarkable that no symmetry hypothesis is made on $\widehat W$. The material symmetry of the model relies entirely on the flow rule $\widehat K$ and is solely due to the internal variable, which in a sense resides on the Lagrangian side of things with our choice of factorization order.
Thus, depending on the flow rule, we can have solid, isotropic or even fluid materials with any stored energy function $\widehat W$, even those that classically describe elastic solids such as the Saint Venant-Kirchhoff or Ciarlet-Geymonat stored energy functions.

For example, the material defined by \eqref{Captain Obvious} is isotropic, even if $\widehat W$ is not. Indeed, its dissipation potential is just $\widehat P_{\mathrm{diss}}(F,\Lambda)=\frac\kappa2\|\Lambda\|^2$, and for all $R\in \SO(3)$,
$$
\widehat P_{\mathrm{diss}}(FR,\Lambda R)=\frac\kappa2\|\Lambda R\|^2=\widehat P_{\mathrm{diss}}(F,\Lambda).
$$
This material is not fluid since $\|\Lambda S^{-T}\|\neq\|\Lambda\|$ in general when $S\notin \SO(3)$ (this can also be checked on $\widehat K$ itself). It is therefore solid.
\end{remark}
  
  Let us now see whether it is actually possible to construct a fluid 3d Maxwell model that is frame-indifferent and satisfies the mechanical part of the Clausius-Planck inequalities. We first characterize all fluid dissipation potentials.
  
  \begin{proposition}A dissipation potential $\widehat P_{\mathrm{diss}}$ is fluid if and only if there exists a fonction $\check P_{\mathrm{diss}}\colon \M_3\times\R_+^*\to\R$ such that
  $$\widehat P_{\mathrm{diss}}(F,\Lambda)=\check P_{\mathrm{diss}}(\Lambda F^T\!,\det F).$$
  \end{proposition}
 
 \begin{proof}Let us be given a dissipation potential giving rise to a fluid material. Given $F_1,F_2\in\M_3^+$ and $\Lambda_1,\Lambda_2\in\M_3$, we see that if  there exists $S\in \SL(3)$ such that $F_1=F_2S$ and $\Lambda_1=\Lambda_2S^{-T}$, then $\det F_1=\det F_2$ and $S=F_2^{-1}F_1$ so that $\Lambda_1F_1^T=\Lambda_2F_2^T$. Conversely, if $\det F_1=\det F_2$ and  $\Lambda_1F_1^T=\Lambda_2F_2^T$, then $S=F_2^{-1}F_1\in\SL(3)$, and $F_1=F_2S$ and $\Lambda_1=\Lambda_2S^{-T}$. It follows that $\widehat P_{\mathrm{diss}}$ is actually a function of $\Lambda F^T$ and $\det F$.
 
 Conversely, given any function $\check P_{\mathrm{diss}}$ as above, if we define $\widehat P_{\mathrm{diss}}$ by $\widehat P_{\mathrm{diss}}(F,\Lambda)=\check P_{\mathrm{diss}}(\Lambda F^T\!,\det F)$, then for all $S\in \SL(3)$
 $$
 \widehat P_{\mathrm{diss}}(FS,\Lambda S^{-T})=\check P_{\mathrm{diss}}(\Lambda S S^{-1}F^T\!,\det (FS))
 =\check P_{\mathrm{diss}}(\Lambda F^T\!,\det F)= \widehat P_{\mathrm{diss}}(F,\Lambda )
 $$
 and we have a fluid material.
 \end{proof}
 
 \begin{remark}We need such dissipation potentials to be frame-indifferent as well, which amounts to requiring that $\check P_{\mathrm{diss}}(NR,J)=\check P_{\mathrm{diss}}(N,J)$ for all $N\in\M_3$, $R\in\SO(3)$ and $J\in\R_+^*$. 
 
 Finally, to make sure that the mechanical part of the Clausius-Planck inequalities is satisfied, the potentials should be convex with respect to $\Lambda$, nonnegative and zero for $\Lambda=0$. An easy example of such a potential satisfying all the above conditions is 
$$
 \widehat P_{\mathrm{diss}}(F,\Lambda)=\frac \kappa2\|\Lambda F^T\|^2
$$
with $\kappa>0$, which thus yields a fluid 3d Maxwell material. For this material, we have
$$
 \frac{\partial  \widehat P_{\mathrm{diss}}}{\partial\Lambda}(F,\Lambda)=\kappa\Lambda F^TF,
 $$
 which results in the flow rule
$$
 \widehat K(F,F_i)=-\frac{\partial  \widehat P_{\mathrm{diss}}}{\partial\Lambda}\Bigl(F,\frac{\partial\widehat A_m}{\partial F_i}(F,F_i)\Bigl)=\kappa F_i^{-T}F^T\frac{\partial\widehat  W}{\partial F_e}(FF_i^{-1})F_i^{-T}F^TF.
$$
 \end{remark}
 
 \begin{remark}
 So far, there were next to no assumptions on the elastic energy $\widehat W$, except to be frame-indifferent. It is thus unlikely that the resulting models, either solid or fluid, would exhibit stress relaxation in all cases. In the case of an isotropic function $\widehat W$, it is not hard to give reasonable sufficient conditions in the particular case of uniform dilatations $F(X,t)= \alpha(t) I$ and $F_i(X,t)=\beta(t) I$, that ensure stress relaxation, either for the first Piolà-Kirchhoff stress vector or more physically for the Cauchy stress vector. We do not pursue in this direction here. 
 \end{remark}
 
 \begin{remark}We have not included temperature in the model, but it can easily be added. If the underlying nonlinearly elastic stored energy function has a natural state $F_0$, \emph{i.e.}, $\frac{\partial\widehat  W}{\partial F_e}(F_0)=0$, then for all $F\in \M_3^+$, $\frac{\partial\widehat A_m}{\partial F_i}(F,F_0^{-1}F)=0$ by equation \eqref{force thermodynamique} and Proposition \ref{CP OK} applies, showing that this a situation where the Clausius-Duhem inequality is equivalent to the Clausius-Planck inequalities if the heat flux law does not depend on $F_i$, an assumption that is fairly reasonable.
 \end{remark}
 
 \subsection{Nonlinear  3d Kelvin-Voigt and 3d generalized Maxwell models }
 The 0d Kelvin-Voigt rheological model consists in a spring and a dashpot in parallel, so that their forces add up. There is thus no internal variable and the model fits well within the general visco-elastic framework without temperature nor internal variables. The natural 3d generalization is thus a special case of our general case. More specifically, we take a frame-indifferent nonlinearly elastic stored energy function $\widehat W$, use it as free energy, \emph{i.e.}, $\widehat A_m(F)=\widehat W(F)$ (again with $\Rho=1$), and choose any constitutive law for the dissipative part of the stress $\widehat T_{\mathrm{Rd}}\colon \M_3^+\times \M_3\to \M_3$ such that
 $$\widehat T_{\mathrm{Rd}}(F,H):H\ge 0,$$
 for instance 
 $$\widehat T_{\mathrm{Rd}}(F,H)=\nu \Sym(HF^{-1})F^{-T}$$
with $\nu>0$, which corresponds to a symmetric Cauchy stress and is obtained from the frame-indifferent dissipation potential $\widehat P_{\mathrm{diss}}(F,H)=\frac{\nu}2\|\Sym(HF^{-1})\|^2$.
 This yields a constitutive law for the stress tensor of the form
 $$
 \widehat T_{\mathrm{R}}(F,H)=\frac{\partial\widehat W}{\partial F}(F)+\widehat T_{\mathrm{Rd}}(F,H).
 $$
 
Here again, sufficient conditions can be given so that in the particular case of uniform dilatations, $F(X,t)= \alpha(t) I$ and $H(X,t)=\alpha'(t) I$, creep---a characteristic feature of the  Kelvin-Voigt  model---does occur.

Another popular rheological model is the 0d generalized Maxwell model which consists in $n+1$ branches connected in parallel, with a spring in branch $0$ and a spring and dashpot in series in each of the other $n$ branches. This setup is easily extended to a 3d framework with an internal variable model. We still have $F\in\M_3^+$ as thermodynamic variable and an internal variable $F_i\in (\M_3^+)^n$. We denote by $F_{i,k}\in\M_3^+$ the $k$-th component of $F_i$. We then consider $n+1$ frame-indifferent nonlinearly elastic stored energy functions $\widehat W_k$, $k=0,\ldots, n$, and define a constitutive law for the free energy by
$$
\widehat A_m(F,F_i)=\widehat W_0(F)+\sum_{k=1}^n\widehat W_k(FF_{i,k}^{-1}).
$$
This yields a constitutive law for the stress of the form
$$
\widehat T_{\mathrm{R}}(F,F_i)=\frac{\partial\widehat  W_0}{\partial F}(F)+\sum_{k=1}^n\frac{\partial\widehat  W_k}{\partial F_e}(FF_{i,k}^{-1})F_{i,k}^{-T}.
$$

A flow rule derived from a frame-indifferent convex dissipation potential will make the model frame-indifferent and satisfying the mechanical part of the Clausius-Planck inequalities. The simplest example of such a potential would be $\widehat P_{\mathrm{diss}}(F,\Lambda)=\frac\kappa2\sum_{k=1}^n\|\Lambda_k\|^2$ with $\kappa>0$ for which the ordinary differential equations for each component of the internal variable decouple,
 $$
\frac{\partial F_{i,k}}{\partial t}(t)=\widehat K_k(F(t),F_{i,k}(t)),
\quad \widehat K_k(F,F_{i,k})=\kappa F_{i,k}^{-T}F^T\frac{\partial\widehat  W_k}{\partial F_e}(FF_{i,k}^{-1})F_{i,k}^{-T}.
$$

Note that material symmetry considerations now involve not only the functions $\widehat K_k$, but $\widehat W_0$ as well. For instance, if the material is to be fluid, then $\widehat W_0$ must be the stored energy function of an elastic fluid, \emph{i.e.}, a function of $J$ only. 

With appropriate assumptions, 3d generalized Maxwell models should be able to exhibit both stress relaxation and creep. We could also mix generalized Maxwell and Kelvin-Voigt together to obtain 3d frame-indifferent models, using thermodynamic variables $F$ and $H$, and internal variables $F_i$ in the fairly obvious fashion.

\subsection{Oldroyd B and Zaremba-Jaumann fluids}\label{fluides complexes}
We conclude this article with two examples of so-called complex fluids, which at first glance do not seem to fit in our general framework, even though they actually do. These fluids are easier to work with in the Eulerian description. Their main characteristic is that the constitutive law for the Cauchy stress is not given by a function of the thermodynamic variables, but by an ordinary differential equation in time, again with often unspoken initial conditions. In the simplest cases, this ordinary differential equation takes the form 
\begin{equation}\label{edo objective generale}
\overset{\diamond}{\sigma}(x,t)=G(\sigma(x,t),d(x,t),\overset{\diamond}{d}(x,t)),
\end{equation}
where $\diamond$ is a differential operator which is of first order in time, and $G$ is some given function. The operator is often---but not always---of the form
\def\Ob{\mathrm{Ob}}
\begin{equation}\label{forme derivee objective}
\overset{\diamond}{\sigma}=\dot \sigma+\Ob(\sigma,h),
\end{equation}
where $\Ob\colon\Sym_3\times \M_3\to\Sym_3$. We recall that $h$ stands for $\nabla_xv$, $d$ for its symmetric part and $w$ for its skew-symmetric part.

In order for such a behavior to be frame-indifferent, the operator $\diamond$ needs to be objective in the following sense.

\begin{definition}An operator $\diamond$ is objective if
$$
\overset{\diamond^*}{\sigma^{\smash{*}}}(x^*,t)=R(t)\overset{\diamond}{\sigma}(x,t)R(t)^T,
$$
for all possible  $\sigma$  and functions $t\mapsto R(t)\in\SO(3)$, using the starred-unstarred notation as before.
\end{definition}

Such an operator is usually called an objective derivative, even though it is not a derivative in the usual technical sense. If the function $G$ is itself frame-indifferent in the sense of 
$$
G(R\sigma R^T,RdR^T, ReR^T)=RG(\sigma,d,e)R^T,
$$
and if the ordinary differential equation has reasonable local uniqueness, then the model will be frame-indifferent.

As is well known, the material derivative $\cdot$, which is a real derivative, is not objective because of the terms involving $\dot R(t)$. There are infinitely many different objective derivatives, of the above form or otherwise. We single out two of the most prominent ones in the literature, the Zaremba-Jaumann derivative, which is the earliest example \cite{Zaremba} and in some sense the simplest one, and the Oldroyd B derivative \cite{Oldroyd}. 
 \begin{definition}The Zaremba-Jaumann derivative is defined by 
$$\overset{\vartriangle}{\sigma}=\dot{\sigma}+\sigma w-w\sigma,$$
and the Oldroyd B derivative by 
$$\overset{\triangledown}{\sigma}=\dot \sigma-h\,\sigma-\sigma h^T.$$
Both are objective derivatives.
\end{definition}
The Oldroyd B derivative is classically used to describe a complex fluid consisting of two components, a polymer and a solvent. The equation for the stress is 
$$
\sigma+\lambda_1\overset{\triangledown}{\sigma}=2\eta(d+\lambda_2\overset{\triangledown}{d}),
$$
where $\eta$, $\lambda_1>\lambda_2$ are strictly positive constants, see \cite{Renardy}. Actually, the Oldroyd B fluid is assumed to be incompressible, so there is also an additional indeterminate pressure term which we do not write as it plays no role in the Clausius-Planck inequality. We adopt here the same equation for a Zaremba-Jaumann fluid, \emph{i.e.},
$$
\sigma+\lambda_1\overset{\vartriangle}{\sigma}=2\eta(d+\lambda_2\overset{\vartriangle}{d}),
$$
again in an incompressible context, see  also \cite{Eiter}.

In a first approach, we perform the Coleman-Noll procedure in this Eulerian, incompressible setting, using only $h$ as thermodynamic variable, with no internal variables, and replacing the constitutive law for the dissipative stress by the differential equation \eqref{edo objective generale}, with a free energy only depending on temperature, so that thermal effects are decoupled from mechanical effects. We skip the details here, but the outcome is that the mechanical part of the Clausius-Planck inequalities reduces to the internal dissipation inequality
$$\sigma:d\ge 0$$
for all arguments and corresponding solutions of the objective differential equation.

In the case of Oldroyd B, if we ignore incompressibility but still with zero free energy so with no elastic pressure, it is fairly easy to construct such arguments and solutions for which, even though $\sigma(0):d(0)\ge 0$, there is a time $t_0$ such that for all $t>t_0$, $\sigma(t):d(t)< 0$, \emph{i.e.}, the second principle is violated. Taking incompressibility into account, we only have numerical evidence of the same violation, see Figure \ref{Oldroyd B c est nul} below. We numerically tested the following example. 

First of all, it is a simple algebraic manipulation to show that the equation
\begin{equation}\label{equation objective generique}
\sigma+\lambda_1\overset{\diamond}{\sigma}=2\eta(d+\lambda_2\overset{\diamond}{d}),
\end{equation}
can be equivalently rewritten as $\sigma=\sigma_s+\sigma_p$, where the subscript $s$ is for solvent and the subscript $p$ is for polymer, with 
\begin{equation}\label{equation objective generique decouplee}
\sigma_s=2\eta_sd,\quad \sigma_p+\lambda_1\overset{\diamond}{\sigma}_p=2\eta_pd,
\end{equation}
where $\eta_s=\frac{\lambda_2}{\lambda_1}\eta$ is the solvent viscosity and  $\eta_p=\bigl(1-\frac{\lambda_2}{\lambda_1}\bigr)\eta$ the polymer viscosity, see \cite{Renardy} for the Oldroyd B case.

 Let $m$ be a randomly chosen  traceless $3\times 3$ matrix. We pick a point $x_0\in E$ and let $v(x,t)=\cos(\omega t)m(x-x_0)$, which amounts to shaking the fluid periodically in time. Since $v(x_0,t)=0$, we thus have
$$\overset{\triangledown}{\sigma}_p(x_0,t)=\frac{\partial\sigma_p}{\partial t}(x_0,t)-\cos(\omega t)(m\sigma_p(x_0,t)+\sigma_p(x_0,t)m^T).$$
We take the values $\eta=1$, $\lambda_2=1$, $\lambda_1=10$, $\omega=0.75$ and an initial polymer stress value $\sigma_p(0)=0$, removing from the notation the point $x_0$, which remains fixed throughout. We then use a standard ode solver to approximate the solution of the $\Sym_3$-valued Cauchy problem
$$
\sigma_p'(t)=\cos(\omega t)(m\sigma_p(t)+\sigma_p(t) m^T)+\frac1{\lambda_1}\bigl(-\sigma_p(t)+2\eta_pd(t)\bigr),
$$
with the above initial value on the time interval $[0,4]$. We then compute and plot $\sigma(t):d(t)$ on the same time interval and obtain the typical kind of evolution portrayed in Figure \ref{Oldroyd B c est nul}, which exhibits strictly negative dissipation for some periods of time. The behavior appears to be quite generic with respect to the choice of numerical values for the constants.

The same numerical computations performed with the Zaremba-Jaumann derivative instead of the Oldroyd B derivative yield the same kind of quantitative behavior for the internal dissipation. Given the rather innocuous ordinary differential equations and the accuracy of standard solvers, we are thus led to very strongly suspect that this form of the second principle is violated by both incompressible Oldroyd B and incompressible Zaremba-Jaumann fluids. This can be informally explained by the fact that such an ordinary differential equation causes the stress to  lag behind the stretching tensor in a sense, so that in a periodic shaking scenario, they may find themselves in opposition of phase at some point in time.
\begin{figure}[htp]
\begin{center}
\includegraphics[scale=.4]{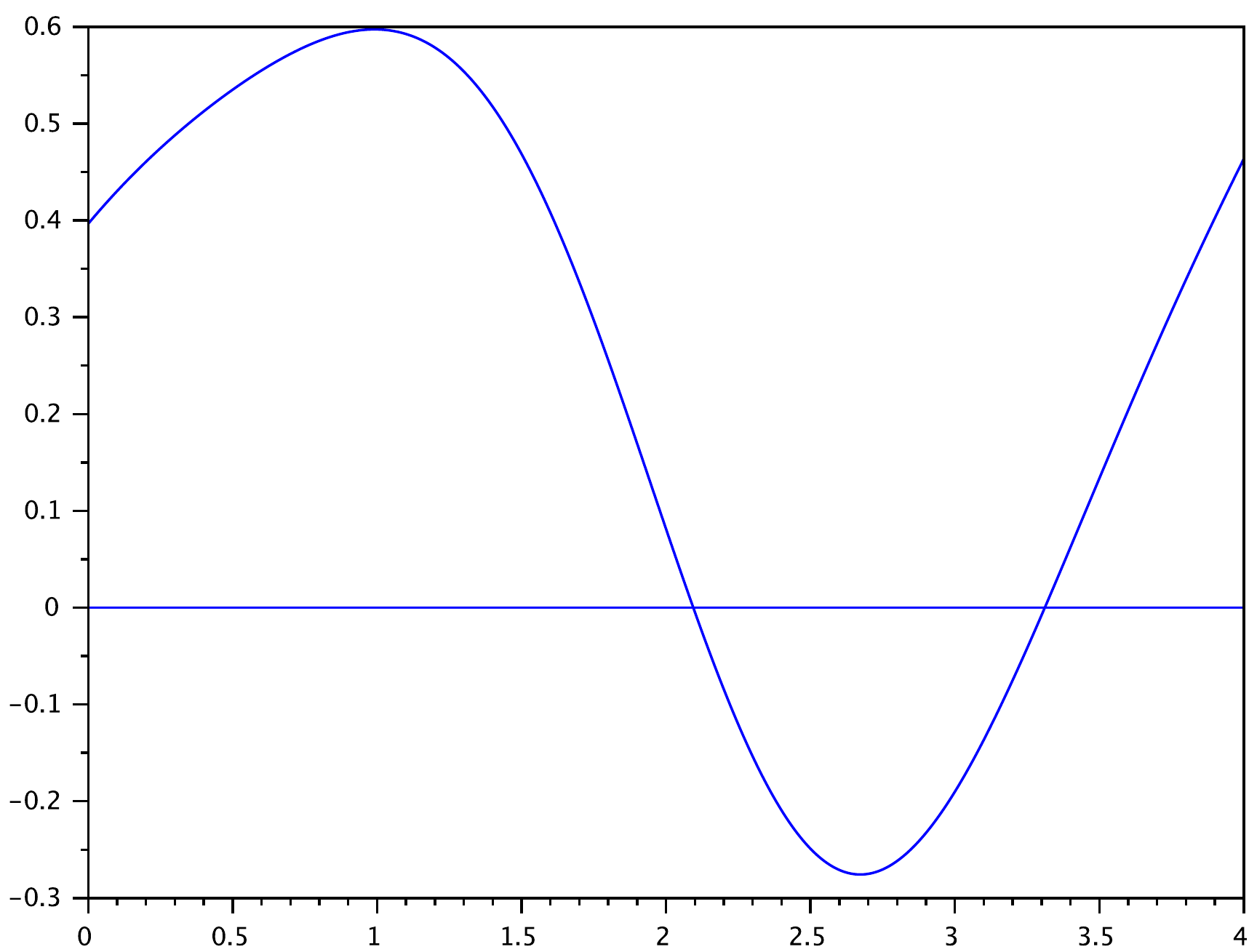}
\caption{Internal dissipation in an Oldroyd B fluid}\label{Oldroyd B c est nul}
\end{center}
\end{figure}

However,  in a second approach, we are going to show that, contrary to appearances, both Oldroyd B and Zaremba-Jaumann fluids can actually be considered as kinematically viscous fluids with an internal variable. They are thus included of our general framework, suitably modified to take incompressibility into account, which is not difficult in the Eulerian description. 

Before going into the specifics of Oldroyd B and Zaremba-Jaumann fluids, let us give a quick rundown of the Coleman-Noll procedure in the Eulerian incompressible case.

Ignoring temperature, we thus have one thermodynamic variable $h$, which is a traceless matrix. We also have two kinds of internal variables $(\pi,\xi)$ and constitutive laws $\widehat a_m(h,\pi,\xi)$ for the free energy and $\widehat \sigma(h,\pi,\xi)$ for the Cauchy stress. As before, we must assume an ordinary differential equation 
$$
\dot\xi=\widehat k(h,\pi,\xi),
$$
and the second principle implies that $\widehat a_m$ only depends on $\xi$ and that the dissipation inequality
\begin{equation}\label{dissipation interne euler var int}
\Bigl(\widehat\sigma(h,\pi,\xi)-pI\Bigr):d-\rho\frac{\partial\widehat a_m}{\partial\xi}(\xi)\cdot\widehat k(h,\pi,\xi)\ge 0
\end{equation}
holds, where $p$ is the indeterminate pressure. Of course, by incompressibility, we have $I:d=0$ and the corresponding term disappears from the dissipation inequality. We can also take $\rho=1$ for the same reason.

We see the same natural decomposition of the (Cauchy) stress
$$
\widehat\sigma(h,\pi,\xi)=\widehat\sigma_{\text{diss}}(h,\pi,\xi)+pI,
$$
into a dissipative part and here an indeterminate pressure part, which would be replaced by an elastic pressure part in the compressible case.

The dissipation potential idea works here too, \emph{i.e.}, a function  $\widehat p_{\mathrm{diss}}\colon \M_3\times\R^m\times\R^k\to\R_+$ convex with respect to its first and third arguments and such that $\widehat p_{\mathrm{diss}}(0,\pi,0)=0$. Then, 
 $$\widehat\sigma_{\text{diss}}(h,\pi,\xi)=\frac{\partial \widehat p_{\mathrm{diss}}}{\partial h}\Bigl(h,\pi,\frac{\partial\widehat a_m}{\partial\xi}(\xi)\Bigr)$$
gives a constitutive law for the dissipative part of the stress and
  $$\widehat k(h,\pi,\xi)=-\frac{\partial  \widehat p_{\mathrm{diss}}}{\partial\lambda}\Bigl(h,\pi,\frac{\partial\widehat a_m}{\partial\xi}(\xi)\Bigr)$$
a flow rule for the internal variable $\xi$, which ensure the mechanical part of the Clausius-Planck inequalities. Let us note that the symmetry of the  Cauchy stress tensor implies that $\widehat p_{\mathrm{diss}}$ only depends on $h$ through $d=\Sym(h)$. We could also discuss frame-indifference issues.

Let us now see how Oldroyd B and Zaremba-Jaumann fluids fit into this picture. We go back to decomposition \eqref{equation objective generique decouplee}.
The idea is to set $\xi=\sigma_p$ and use the constitutive law $\widehat\sigma_{\text{diss}}(h,\xi)= 2\eta_sd+\xi$, together with the ordinary differential equation $\xi+\lambda_1\overset{\diamond}{\xi}=2\eta_pd$, which assumes the requisite forms
$$
\dot\xi=h\xi+\xi h^T+\frac1{\lambda_1}\bigr({-}\xi+2\eta_pd\bigl),
$$
\emph{i.e.},
\begin{equation}\label{flow rule OB}
\widehat k(h,\xi)=h\xi+\xi h^T+\frac1{\lambda_1}\bigr({-}\xi+2\eta_pd\bigl),
\end{equation}
for Oldroyd B (there is no $\pi$ kind of internal variable), and 
$$
\dot\xi=w\xi-\xi w+\frac1{\lambda_1}\bigr({-}\xi+2\eta_pd\bigl),
$$
\emph{i.e.},
\begin{equation}\label{flow rule ZJ}
\widehat k(h,\xi)=w\xi-\xi w+\frac1{\lambda_1}\bigr({-}\xi+2\eta_pd\bigl),
\end{equation}
for Zaremba-Jaumann. It is worth mentioning that any complex fluid model based on such a differential equation as \eqref{equation objective generique}, with an objective derivative of the form \eqref{forme derivee objective} fits equally well in this mould.

On a side note, in both Oldroyd B and Zaremba-Jaumann cases, $\widehat k$ does not depend on $h$ only through $d$ (take $h$ skew-symmetric), which precludes the existence of a dissipation potential.

The question now is, is it possible to choose free energies $\widehat a_m(\xi)$ such that inequality \eqref{dissipation interne euler var int} is satisfied by either one of the two models?

The Oldroyd B case in answered in the negative by the following proposition.

\begin{proposition}There exists no $C^2$ function $\widehat a_m$ such that the dissipation inequality \eqref{dissipation interne euler var int} is satisfied by the Oldroyd B fluid.
\end{proposition}
\begin{proof}In the Oldroyd B model, we have $\tr(\sigma_s)=0$, but the trace of $\sigma_p=\xi$ is not constrained to any given value. We can thus take any symmetric matrix as initial value or in the ordinary differential equation. Let us accordingly assume that there exists a $C^2$ function $\widehat a_m\colon \Sym_3\to\R$ such that
\begin{equation}\label{un clausius-planck de plus}
\sigma:d+z(\xi):\widehat k(h,\xi)\ge 0\text{ with }z(\xi)=-\frac{\partial \widehat a_m}{\partial\xi}(\xi)
\end{equation}
and $\widehat k$ is given by \eqref{flow rule OB}, for all $h\in\M_3$ such that $\tr h=0$ and all $\xi\in\Sym_3$.
Inequality \eqref{un clausius-planck de plus} is first expanded as
$$
2\eta_s\|d\|^2+\xi:d+z(\xi):\widehat k(h,\xi)\ge 0.
$$
We remark that the left-hand side is a  polynomial of degree at most $2$ in  $h$. In particular, the transformation $h\to sh$ with $s\in \R$ shows that 
\begin{equation}\label{trois clausius-planck de plus}
2\eta_s\|d\|^2s^2+\biggl(\xi:d+z(\xi):\Bigl(h\xi+\xi h^T+\frac{2\eta_p}{\lambda_1}d\Bigr)\biggr)s-\frac1{\lambda_1}z(\xi):\xi\ge 0
\end{equation}
for all $s\in\R$ and all $\xi$ and $h$. In the sequel, we let $\alpha=\frac{\eta_p}{\lambda_1}$ and $\beta=\frac{8\eta_s}{\lambda_1}$. 

Setting $s=0$, we obtain a first necessary condition
\begin{equation}\label{une CN de plus}
z(\xi):\xi\le 0
\end{equation}
for all $\xi$. Then, there is a discussion according to whether
\begin{itemize}
\item $d\neq0$, in which case \eqref{trois clausius-planck de plus} holds if and only if the discriminant is nonpositive, 
\begin{equation}\label{deux CN de plus A}
\bigl(\xi:d+z(\xi):(h\xi+\xi h^T+2\alpha d)\bigr)^2+\beta\|d\|^2z(\xi):\xi\le 0,
\end{equation}
\item $d=0$, in which case  \eqref{trois clausius-planck de plus} holds if and only if
\begin{equation}\label{deux CN de plus B}
z(\xi):(h\xi+\xi h^T)= 0.
\end{equation}
Now if $d=0$, then $h$ is skew-symmetric, in particular traceless, and the previous condition reads $z(\xi):(h\xi-\xi h)=-2\bigl(\xi z(\xi)\bigr):h= 0$, for all $h\in\Skew_3$. We deduce from this a second necessary condition,
\end{itemize}
\begin{equation}\label{trois CN de plus}
\xi z(\xi)\in \Sym_3\text{ \emph{i.e.} $\xi$ and $z(\xi)$ commute.}
\end{equation}

%On repasse au cas général $d\neq 0$, $\tr d=0$, c'est-à-dire \eqref{deux CN de plus A}. On voit que pour $\xi=0$, $z(0)=0$ convient, entre autres. Par contre, si $\xi\neq0$, alors nécessairement $z(\xi)\neq 0$, en prenant $d$ tel que $\xi:d\neq0$. 
Let us go back to \eqref{deux CN de plus A}.
First of all,
$$
z(\xi):(h\xi+\xi h^T)=2\bigl(\xi z(\xi)\bigr):h=2\bigl(\xi z(\xi)\bigr):d
$$
because $\xi z(\xi)$ is symmetric by \eqref{trois CN de plus}. This provides a new equivalent version of \eqref{deux CN de plus A},
\begin{equation}\label{deux CN de plus A bis}
\bigl((\xi+2\xi z(\xi)+2\alpha z(\xi)):d\bigr)^2+\beta\|d\|^2z(\xi):\xi\le 0.
\end{equation}
By the Cauchy-Schwarz inequality, the worst case scenario for the first term in the left-hand side of  \eqref{deux CN de plus A bis} is
$$d_0(\xi)=\xi+2\xi z(\xi)+2\alpha z(\xi)-\frac13\tr\bigl(\xi+2\xi z(\xi)+2\alpha z(\xi)\bigr)I,$$
from which we get another necessary condition,
\begin{equation}\label{deux CN de plus A ter}
\|d_0(\xi)\|^2\bigl(1+\beta z(\xi):\xi\bigr)\le 0.
\end{equation}
The second term in the product is nonpositive if and only if
$$
z(\xi):\xi\le -\frac{1}{\beta}<0,
$$
which is impossible in a neighborhood of $\xi=0$ because $z$ is continuous. From \eqref{deux CN de plus A ter}, we therefore have
$$d_0(\xi)=0,$$
that is to say
\begin{equation}\label{relation pour z}
\xi+2\xi z(\xi)+2\alpha z(\xi)=\mu(\xi)I,
\end{equation}
a neighborhood of $0$ with $\mu(\xi)=\frac13\tr\bigl(\xi+2\xi z(\xi)+2\alpha z(\xi)\bigr)$.
Conversely, if relation \eqref{relation pour z} is satisfied for some function $\mu$, then \eqref{deux CN de plus A bis} holds, since $I:d=\tr d=0$.

We thus see that in a neighborhood of $0$, 
\begin{equation}\label{forme de z}
z(\xi)=\frac12(\xi+\alpha I)^{-1}(\mu(\xi)I-\xi),
\end{equation}
where $\mu$ is a so far arbitrary real-valued function defined on this neighborhood. 

First of all, $z$ defined by \eqref{forme de z} commutes with $\xi$, hence \eqref{trois CN de plus} is satisfied. Secondly, this function $z$ must satisfy \eqref{une CN de plus}. We have
$$
z(\xi)
=\frac1{2\alpha}(I+o(1))(\mu(0)I+o(1))=\frac{\mu(0)}{2\alpha}I+o(1)
$$
so that by \eqref{une CN de plus},
$$0\ge z(\xi):\xi=\frac{\mu(0)}{2\alpha}\tr\xi+o(\|\xi\|),$$
which implies that 
\begin{equation}\label{mu de zero}
\mu(0)=0.
\end{equation}

To reach a contradiction, we now use the fact that $z$ is a gradient.
In order to simplify the expressions, we remark that
$$
z(\xi)=-\frac12I+\frac{\mu(\xi)+\alpha}2(\xi+\alpha I)^{-1},
$$
and we perform the change of variables $\zeta=\xi+\alpha I$ and change of unknown function $\nu(\zeta)=\mu(\xi)+\alpha$, so that $\nu(\alpha I)=\alpha$, for $\zeta$ in a neighborhood of $\alpha I$, and
$$
z(\xi)=-\frac12I+\frac{\nu(\zeta)}2\zeta^{-1},
$$
by \eqref{forme de z}. The first term in the right-hand side is the gradient of $\zeta\mapsto-\frac12\tr\zeta$, so we just need to focus on 
$$
\zeta\mapsto y(\zeta)=\nu(\zeta)\zeta^{-1},
$$
which must also be a gradient. Since $y$ has the same smoothness as $\nu$, we see that $\nu$ is $C^1$. Furthermore, the zero curl condition must be satisfied,
 $$
 \frac{\partial y_{ij}}{\partial\zeta_{kl}}(\zeta)=\frac{\partial y_{kl}}{\partial\zeta_{ij}}(\zeta),
 $$
 for all indices $i,j,k,l$, and matrices $\zeta$ in a neighborhood of $\alpha I$. We take matrices of the form
 \begin{equation}\label{forme des matrices}
 \zeta=\left(\begin{array}{c|c}
  \bar \zeta & 0 \\ 
  \hline
  0 & \zeta_{33}
 \end{array}\right)\text{ with }\bar\zeta\in\Sym_2,
\end{equation}
for which 
$$ y(\zeta)=
\nu(\zeta)\left(\begin{array}{c|c}
  \frac1{\strut\det\bar\zeta}
\begin{pmatrix}
\zeta_{22}&-\zeta_{12}\\-\zeta_{12}&\zeta_{11}
\end{pmatrix}& 0 \\ 
  \hline
  0 & \zeta_{33}^{-1}
 \end{array}\right).
$$

We only write the derivatives that we will use:
\begin{align*}
\frac{\partial y_{11}}{\partial\zeta_{12}}(\zeta)&=\frac1{\det\bar\zeta}\Bigl(\frac{\partial\nu}{\partial\zeta_{12}}(\zeta)+2\frac{\nu(\zeta)}{\det\bar\zeta}\zeta_{12}\Bigr)\zeta_{22}, \\
\frac{\partial y_{11}}{\partial\zeta_{33}}(\zeta)&=\frac{\zeta_{22}}{\det\bar\zeta}\frac{\partial\nu}{\partial\zeta_{33}}(\zeta),
\end{align*}
then
 \begin{align*}
\frac{\partial y_{12}}{\partial\zeta_{11}}(\zeta)&=-\frac1{\det\bar\zeta}\Bigl(\frac{\partial\nu}{\partial\zeta_{11}}(\zeta)-\frac{\nu(\zeta)}{\det\bar\zeta}\zeta_{22}\Bigr)\zeta_{12},\\
\frac{\partial y_{12}}{\partial\zeta_{33}}(\zeta)&=-\frac{\zeta_{12}}{\det\bar\zeta}\frac{\partial\nu}{\partial\zeta_{33}}(\zeta),
\end{align*}
and finally
 \begin{align*}
\frac{\partial y_{33}}{\partial\zeta_{11}}(\zeta)&=\frac1{\zeta_{33}}\frac{\partial\nu}{\partial\zeta_{11}}(\zeta),\\
\frac{\partial y_{33}}{\partial\zeta_{12}}(\zeta)&=\frac1{\zeta_{33}}\frac{\partial\nu}{\partial\zeta_{12}}(\zeta).
\end{align*}

The relation $\frac{\partial y_{11}}{\partial\zeta_{12}}=\frac{\partial y_{12}}{\partial\zeta_{11}}$ reads
$$
\Bigl(\frac{\partial\nu}{\partial\zeta_{12}}(\zeta)+2\frac{\nu(\zeta)}{\det\bar\zeta}\zeta_{12}\Bigr)\zeta_{22}
=-\Bigl(\frac{\partial\nu}{\partial\zeta_{11}}(\zeta)-\frac{\nu(\zeta)}{\det\bar\zeta}\zeta_{22}\Bigr)\zeta_{12}
$$
or equivalently
\begin{equation}\label{1ere relation infecte}
\frac{\partial\nu}{\partial\zeta_{12}}(\zeta)\zeta_{22}+\frac{\partial\nu}{\partial\zeta_{11}}(\zeta)\zeta_{12}=
-\frac{\nu(\zeta)}{\det\bar\zeta}\zeta_{12}\zeta_{22}.
\end{equation}

The relation $\frac{\partial y_{33}}{\partial\zeta_{11}}=\frac{\partial y_{11}}{\partial\zeta_{33}}$ reads
\begin{equation}\label{4eme relation infecte}
\frac1{\zeta_{33}}\frac{\partial\nu}{\partial\zeta_{11}}(\zeta)=
\frac{\zeta_{22}}{\det\bar\zeta}\frac{\partial\nu}{\partial\zeta_{33}}(\zeta).
\end{equation}

The relation $\frac{\partial y_{33}}{\partial\zeta_{12}}=\frac{\partial y_{12}}{\partial\zeta_{33}}$ reads
\begin{equation}\label{5eme relation infecte}
\frac1{\zeta_{33}}\frac{\partial\nu}{\partial\zeta_{12}}(\zeta)=-
\frac{\zeta_{12}}{\det\bar\zeta}\frac{\partial\nu}{\partial\zeta_{33}}(\zeta).
\end{equation}

Replacing \eqref{4eme relation infecte} and \eqref{5eme relation infecte} into \eqref{1ere relation infecte}, we deduce that $\nu(\zeta)\zeta_{12}\zeta_{22}=0$, hence by continuity, $\nu(\zeta)=0$ for matrices $\zeta$ of the form \eqref{forme des matrices} in a neighborhood of $\alpha I$. This contradicts $\nu(\alpha I)=\alpha$, viz.\ \eqref{mu de zero}.\end{proof}

We have not stressed regularity issues thus far, but it is highly unlikely that allowing for a less regular function $\widehat a_m$ would alleviate the problem.

The situation for the Zaremba-Jaumann fluid with respect to the second principle is much better.

\begin{proposition}The Zaremba-Jaumann fluid satisfies the second principle with the choice $\widehat a_m(\xi)=\frac{\lambda_1}{4\eta_p}\|\xi\|^2$.
\end{proposition}

\begin{proof}
For this choice, we have $z(\xi)=-\frac{\lambda_1}{2\eta_p}\xi$ and the Clausius-Planck inequality \eqref{un clausius-planck de plus} to be satisfied becomes
\begin{equation}\label{un clausius-planck de plus en plus}
2\eta_s\|d\|^2+\xi:d-\frac{\lambda_1}{2\eta_p}\xi:\Bigl(w\xi-\xi w+\frac1{\lambda_1}
(-\xi+2\eta_p d)\Bigr)\ge 0.
\end{equation}
We first remark that $\xi^2$ is symmetric and $w$ is skew-symmetric, thus
$$
0=(\xi^2):w=\tr(\xi\xi w)=\xi:(\xi w)=\tr(\xi w\xi)=\xi:(w\xi),
$$
so that the terms involving $w$ in \eqref{un clausius-planck de plus en plus} all vanish. The left-hand side of \eqref{un clausius-planck de plus en plus} thus reduces to $2\eta_s\|d\|^2+\frac1{2\eta_p}\|\xi\|^2$, which is always nonnegative.
\end{proof}

\begin{remark}
There are infinitely many different choices of $\widehat a_m(\xi)$ that make this left-hand side nonnegative, which can be described in detail. However, the choice of a specific free energy  should be based on physical grounds, not on the mathematical fact that it can compensate for $\sigma:d$ becoming strictly negative in time just because $\sigma$ obeys the ordinary differential equation given by any chosen objective derivative. We feel it is nonetheless significant that such a compensation is impossible for the Oldroyd B fluid, whereas there are many mathematical possibilities for a Zaremba-Jaumann fluid. 

Let us point out that we are dealing here with the second principle in the form of the Clausius-Planck inequalities. It is not completely ruled out that an Oldroyd B fluid complemented with a heat flux that depends on $d$ and $\xi=\sigma_p$ could still satisfy the Clausius-Duheim inequality. We do not pursue in this direction here since such a heat flux would presumably be hard to justify on physical grounds.
\end{remark}

%The acknowledgments section should not be numbered.
\section*{Acknowledgments}
We would like to thank the organizers of the conference AMES 2022, where a preliminary version of this work was presented. This conference in honor of Prof.\ Philippe Destuynder was a much appreciated friendly gathering  to celebrate his deep scientific achievements.

%%%%%%%%%%%%%%%%%%%%%%%%%%%%%%%%%%%%%%%%%%%%%%%%%%%%%%
%          7. REFERENCES SECTION
%%%%%%%%%%%%%%%%%%%%%%%%%%%%%%%%%%%%%%%%%%%%%%%%%%%%%%

%       READ THIS SECTION CAREFULLY

% Each of the references below MUST be cited in your article above. Do not include references that are not cited in your article.

% Follow the examples below carefully. We strongly suggest that you copy and paste your reference information directly into our examples.

% List all references in alphabetical order according to the first author's last name.

% Verify each URL works correctly and can be accessed properly. Your URL links should be to reputable websites. The command line for a website link begins with: \url{ }

% Do not add MR or DOI numbers to your references. AIMS production staff will add this information.

% Using BibTex is not recommended but can be handled.

\end{document}